\documentclass[7pt]{article}
\usepackage[utf8]{inputenc}
\usepackage[margin=0.84in]{geometry}

\usepackage{chngpage}

\title{Nonparametric Tests for Treatment Effect Heterogeneity in Observational Studies}
\author{Maozhu Dai, Weining Shen, Hal S. Stern}
\date{}

\usepackage{etoolbox}
\usepackage{enumitem}
\usepackage{graphicx}
\usepackage{adjustbox}
\usepackage{float}
\usepackage{setspace}
\usepackage{amsmath}
\usepackage{subfigure}
\usepackage{amssymb}
\usepackage{amsthm}
\usepackage{float}
\usepackage{rotating}
\usepackage{blkarray}
\usepackage{longtable}
\usepackage{caption}
\usepackage{siunitx}
\usepackage{mathrsfs}
\usepackage[noadjust]{cite}
\usepackage{geometry}
\usepackage{algorithm}
\usepackage{xcolor}
\usepackage{comment}
\usepackage{hyperref}
\usepackage{threeparttable}
 
\usepackage{booktabs}
\usepackage{textgreek}

\usepackage{natbib}
\setcitestyle{authoryear,open={(},close={)}}
 
 \newcommand{\citeg}[1]{\citep[e.g.,][]{#1}}

\theoremstyle{plain}

\newtheorem{theorem}{Theorem}

\begin{document}

\maketitle







\section*{\hfil Abstract \hfil}
We consider the problem of testing for treatment effect heterogeneity in observational studies, and propose a nonparametric test based on multisample U-statistics. To account for potential confounders, we use reweighted data where the weights are determined by estimated propensity scores. The proposed method does not require any parametric assumptions on the outcomes and bypasses the need for modeling the treatment effect for each study subgroup. We establish the asymptotic normality for the test statistic, and demonstrate its superior numerical performance over several competing approaches via simulation studies. Two real data applications including an employment program evaluation study and a mental health study of China's one-child policy are also discussed.

\vspace{6pt}
\noindent
\textit{Keywords}: Causal inference, observational study, reweighting, subgroup analysis, U-statistics.

\def\thefigure{\arabic{figure}}
\def\thetable{\arabic{table}}

\renewcommand{\theequation}{\thesection.\arabic{equation}}

\fontsize{10}{13pt plus.8pt minus .6pt}\selectfont

\section{Introduction}

Understanding treatment effect heterogeneity has attracted a great deal of attention in various research areas, including social sciences \citep{bitler2006mean,feller2009beyond}, health care \citep{kent2016risk,ginsburg2009genomic} and criminology \citep{na2015importance,pate1992formal}. It is now well recognized that ``one size does not fit all" in many disease studies since subjects with different characteristics may respond quite differently to the same treatment. To better account for patient heterogeneity while evaluating the treatment effect and providing accurate personalized treatment recommendation, subgroup analysis \citep{byar1985assessing} has been commonly used to identify subpopulations among subjects and examine the localized treatment effects within subpopulations. In some studies, subjects may be divided into several strata based on baseline characteristics that are expected to be associated with treatment effects, and recommendations are made based on inference conducted within each stratum. However, this procedure is beneficial only if there is enough evidence showing the existence of treatment effect heterogeneity across those strata; otherwise  we simply lose information and statistical power by conducting stratum-specific analysis.

There is an emerging literature on developing hypothesis testing approaches for examining treatment effect heterogeneity \citeg{chang2015nonparametric, ding2016randomization,hsu2017consistent} under different definitions of heterogeneity and different modeling assumptions. In this paper we focus on testing whether the average treatment effects across multiple pre-specified subpopulations are identical to each other. The earliest work towards this goal was the likelihood ratio test (LRT) developed by \cite{gail1985testing} under  normality assumptions for the stratum-specific treatment effect estimates. Regression methods have also been considered, where the heterogeneity of treatment effects is tested by examining interaction terms between treatment assignment and potential effect modifiers \citep{krishnan2003smoking}. More recently, several nonparametric approaches have been proposed in the literature. \cite{crump2008nonparametric} proposed a test based on sieve estimation for treatment effects. This method was later generalized by \cite{sant2016nonparametric} to test for heterogeneity in duration outcomes under endogenous treatment assignment. More recently, \cite{dai2020u} proposed a U-statistic-based test (U test) which does not require estimating stratum-specific treatment effects. 
Compared to the LRT and other parametric tests, the nonparametric tests in general require weaker modeling assumptions on the outcome distributions. However, they still either require specifying a model for estimating the treatment effects \citep{crump2008nonparametric,sant2016nonparametric}, or only consider situations where baseline covariates are well balanced within each stratum \citep{dai2020u}. Motivated by these observations, we propose a nonparametric test that bypasses the need for estimating treatment effects while still being applicable to observational studies where there exist confounding variables that need to be addressed.

In this paper, we focus on testing the equality of the average treatment effects across multiple strata while adjusting for potential confounding variables in observational studies. We propose a new testing procedure based on an adjusted four-sample U-statistic that can be viewed as a weighted version of the original U-statistic developed by \cite{dai2020u}. Assuming the strata are mutually independent, the main idea is to first construct an adjusted U-statistic for comparing the treatment effects between two strata, and then formulate an overall test statistic as a function of those pairwise adjusted U-statistics. For each stratum, the weights in the adjusted U-statistic are carefully chosen by covariate matching and propensity score estimation \citep{li2018balancing} such that the baseline covariate distributions for both the treatment and control groups are the same as the marginal distribution for the target population. To derive the asymptotic distribution for the proposed test, we find the main challenge is that our adjusted U-statistic no longer belongs to the generalized U-statistic family, therefore classical projection theory is not directly applicable. To solve this problem, we use the idea in \cite{satten2018multisample}, which studies 
adjusted two-sample U-statistics, to obtain an asymptotic normality result.  Based on the derived asymptotic theory, we then conduct several numerical studies to compare the performance of our proposed test with that of the LRT \citep{gail1985testing} and the unadjusted U test \citep{dai2020u}. Numerical results confirm the excellent operating characteristics for the proposed method even under propensity score model misspecification, and also clearly demonstrate the advantage of our method over the LRT and the unadjusted U test when the data is generated from a non-Gaussian distribution or the baseline covariates are not well balanced.

The remainder of the paper is structured as follows. In Section~\ref{sec:background}, we provide a review of the U test that assesses  treatment effect heterogeneity across strata with balanced baseline covariates. In Section~\ref{sec:method}, we introduce our adjusted U test for treatment effect heterogeneity that allows for the existence of confounding variables. In Section~\ref{sec:Simulation}, we conduct  simulation studies to demonstrate the asymptotic validity and efficiency of the adjusted U test, and also explore the impact of model misspecification. In Section~\ref{sec:caseStudy}, we further demonstrate the use of our method by two case studies, including an employment program evaluation study in labor economics, and another study on the evaluation of China's one-child policy on children's mental health. We conclude with some remarks in Section~\ref{seq:discussion}. Some additional plots and tables can be found in Supplementary Material.

\section{Review of Unadjusted U-Statistic-Based Test for Treatment Effect Heterogeneity}
\label{sec:background}
\cite{dai2020u} (hereafter DS) proposed a U-statistic-based test (U test) to assess the consistency of average treatment effects in several independent strata, assuming there are no confounding variables. Compared to its parametric counterpart, the Likelihood Ratio Test (LRT) introduced by \cite{gail1985testing}, their proposed U test can have a significant improvement in power especially when the outcomes are deviating far away from a normal distribution. Since the method we propose in this paper is based on their U test, we start with a review of their method.

Assume there are $S$ strata. Within each stratum $s$ $(s\in \{1,...,S\})$, let $\tau_s$  be the additive treatment effect,  $Y_s^t = \{Y_{si}^t, i = 1,...,n_s^t\}$ be the outcomes of subjects in the treatment group, and $Y_s^c = \{Y_{si}^c, i = 1,...,n_s^c\}$ be the outcomes of subjects in the control group. The total sample size across all strata is denoted as $N = \sum\limits_{s =1}^S(n_s^t+n_s^c)$. Two assumptions are made in DS: (1) the outcomes $(Y_1^t, \cdots, Y_S^t, Y_1^c, \cdots, Y_S^c)$ are mutually independent; and (2) there exist positive constants $0<\lambda_s^\omega<1$ for every $s\in\{1,...,S\}$ and $\omega \in \{t,c\}$ such that $\frac{n_s^\omega}{N}\rightarrow \lambda_s^\omega$ as $N \rightarrow \infty$.

To test for treatment effect heterogeneity across all strata, DS considers the null hypothesis that the difference in the potential outcomes follow the same distribution across all strata, i.e., $H_0$: $Y_s^t-Y_s^c$ are identically distributed for every $s \in\{1, \cdots, S\}$, and the alternative hypothesis is that at least two of those distributions are not the same. When $Y_s^t-Y_s^c$ $(s = 1,\cdots, S)$ follow a common distribution up to a location shift, or $Y^t_s$ and $Y^c_s$ follow the same distribution up to a stratum-specific location shift within each stratum $s$ $(s = 1, \cdots, S)$, the hypotheses are equivalent to $H_0: \tau_1 = ... = \tau_S$ versus $H_a:$ at least two of them are not equal, where $\tau_s = \text{E}(Y_s^t) - \text{E}(Y_s^c)$. More discussions about these hypotheses can be found in Section 3.3 of DS. The test statistic is constructed by combining all pairwise U-statistics that compare treatment effects in two strata. To compare the treatment effects in the first two strata,  a four-sample U-statistic is constructed as 
\begin{equation}
     U^{(1,2)} = \frac{1}{n_{1}^tn_{1}^cn_{2}^tn_{2}^c}\sum_{i=1}^{n_{1}^t}\sum_{j=1}^{n_{1}^c}\sum_{k=1}^{n_{2}^t}\sum_{l=1}^{n_{2}^c} \phi^{(1,2)}(i,j,k,l),
     \label{eq:U12}
\end{equation}
where the kernel function $\phi^{(1,2)}(i,j,k,l) = I(Y_{1i}^t-Y_{1j}^c<Y_{2k}^t-Y_{2l}^c)+\frac{1}{2}I(Y_{1i}^t-Y_{1j}^c=Y_{2k}^t-Y_{2l}^c)$. The latter term is used to account for possible ties for discrete distributions. Although DS focuses on additive treatment effect, other forms of treatment effects, such as the ratio of outcomes between different treatment groups, can also be incorporated. DS shows that
\begin{equation}
       \sqrt{N}(U^{(1,2)}-\theta^{(1,2)}) \stackrel{D}{\longrightarrow} \mathcal{N}(0,\sigma^2_{1,2}), ~~~~ \text{when } N\rightarrow \infty,
\end{equation}
where $\sigma^2_{1,2} = \frac{1}{\lambda_1^t}\text{Var}(h_1^{t,(1,2)}(Y_1^t))+
\frac{1}{\lambda_1^c}\text{Var}(h_1^{c,(1,2)}(Y_1^c))+\frac{1}{\lambda_2^t}\text{Var}(h_2^{t,(1,2)}(Y_2^t))+
\frac{1}{\lambda_2^c}\text{Var}(h_2^{c,(1,2)}(Y_2^c))$ is the asymptotic variance of $\sqrt{N}U^{(1,2)}$, and $h_s^{\omega,(1,2)}(x) = \text{E}[\phi^{(1,2)}(1,1,1,1)|Y_{s1}^\omega = x]-\theta^{(1,2)}$ for $s \in\{1,2\}$ and $ \omega\in\{t,c\}$. Under the null hypothesis that the difference of potential outcomes are identically distributed across strata, the expectation of $\phi^{(1,2)}(i,j,k,l)$ is $\frac{1}{2}$, thus $\theta^{(1,2)} \stackrel{\Delta}{=} E(U^{(1,2)})$ is also $\frac{1}{2}$.

With $S$ strata, all pairwise U-statistics $U^{(p,q)}$ $(1\leq p < q \leq S)$ can be constructed in the exactly same way. Specifically, for every pair of $(p,q)$, we can define $U^{(p,q)}$, $\theta^{(p,q)}$ and $h_s^{\omega,(p,q)}$ $(\omega\in\{t,w\}, s\in\{p,q\})$ similarly with $U^{(1,2)}$, $\theta^{(1,2)}$ and $h_s^{\omega,(1,2)}$ by replacing $(1,2)$ with $(p,q)$. Under the assumption that $\frac{n_s^\omega}{N}\rightarrow \lambda_s^\omega$ $(0<\lambda_s^\omega<1)$ as $N\rightarrow\infty$ for $s\in\{1,\cdots,S\}$ and $\omega\in\{t,c\}$, DS shows that
\begin{equation}
\sqrt{N} (U^{(1,2)}-\theta^{(1,2)}, U^{(1,3)}-\theta^{(1,3)}, \cdots, U^{(S-1,S)}-\theta^{(S-1,S)})^T \stackrel{D}{\longrightarrow} \mathcal{N}(0, \Sigma), ~~~~ \text{when } N\rightarrow \infty,
\end{equation}
where $\Sigma = \frac{1}{\lambda_1^t}\Sigma_1^t+\frac{1}{\lambda_1^c}\Sigma_1^c+\dots+\frac{1}{\lambda_S^t}\Sigma_S^t+\frac{1}{\lambda_S^c}\Sigma_S^c$ and $\Sigma_s^\omega$ is the covariance matrix of\\ $\left(\Tilde{h}_s^{\omega,(1,2)}(Y_s^\omega),\Tilde{h}_s^{\omega,(1,3)}(Y_s^\omega),\cdots,\Tilde{h}_s^{\omega,(S-1,S)}(Y_s^\omega)\right)$ for all $s\in\{1,
\cdots,S\}$ and $\omega \in \{t,c\}$. Here $\Tilde{h}_s^{\omega,(p,q)}(x) = h_s^{\omega,(p,q)}(x) I(s = p \text{ or } s = q)$.

To apply this method, $\Sigma$ is estimated by a weighted average of $\Sigma_s^\omega$ $(s\in\{1,...,S\}, \omega\in\{t,c\})$, and $\Sigma_s^\omega$ can be estimated by the corresponding sample covariance matrix. As $\Tilde{h}$ terms are unknown, they need to be estimated as well. Though $h_s^{\omega,(p,q)}(x) = \text{E}[\phi^{(p,q)}(i,j,k,l)|Y_s^\omega = x] - \theta^{(p,q)}$, the constant term $\theta^{(p,q)}$ can be ignored when calculating the covariance matrices.
So they take the method-of-moment estimator for the expectation term $\text{E}[\phi^{(p,q)}(i,j,k,l)|Y_s^\omega = x]$ as the estimator of $h_s^{\omega,(p,q)}(x)$. For instance, the estimator of $h_1^{t,(1,2)}(x)$ is $\hat{h}_1^{t,(1,2)}(x) = \frac{1}{n_{1}^cn_{2}^tn_{2}^c}\sum\limits_{j=1}^{n_{1}^c}\sum\limits_{k=1}^{n_{2}^t}\sum\limits_{l=1}^{n_{2}^c}I(x-Y_{1j}^c<Y_{2k}^t-Y_{2l}^c)$. Similar calculation is repeated for all other $h$ functions, and then used for computing the sample covariance $\hat{\Sigma}_s^\omega$ $(s\in\{1,\cdots,S\}, \omega\in\{t,c\})$, which leads to the final estimator of $\Sigma$ as $\hat{\Sigma} = \frac{1}{\lambda_1^t}\hat{\Sigma}_1^t+\frac{1}{\lambda_1^c}\hat{\Sigma}_1^c+\dots+\frac{1}{\lambda_S^t}\hat{\Sigma}_S^t+\frac{1}{\lambda_S^c}\hat{\Sigma}_S^c$.


To test the null hypothesis $H_0: \theta = \frac{1}{2}\mathbf{1}_{S(S-1)/2}$, where $\theta = (\theta^{(1,2)},\theta^{(1,3)},...,\theta^{(S-1,S)})^T$, DS focuses on a one-dimensional overall test statistic $U_h = N\cdot \sum\limits_{1\leq p<q\leq S}(U^{(p,q)}-\frac{1}{2})^2$. Though the asymptotic reference distribution of $U_h$ does not have an analytic form, it can be approximated by simulation, that is, after generating a large number of independent samples $\{r_1,\cdots, r_L\}$ from $\mathcal{N}(0, \hat{\Sigma})$, the empirical distribution of $\{||r_1||^2,\cdots, ||r_L||^2\}$ approximates the asymptotic reference distribution of $U_h$.


\section{Adjusted U Test of Treatment Effect Heterogeneity}
\label{sec:method}
The test described in Section~\ref{sec:background} can only be used in situations where all baseline covariates are well balanced between different treatment groups in each stratum, e.g., stratified randomized experiments. In observational studies, directly applying that method may lead to misleading conclusions due to the  existence of potential confounding variables. Even in the situations where the strata are constructed based on propensity scores, which is the probability of getting treatment \citep{rosenbaum1983central}, in hope of balancing baseline covariates \citep{xie2012estimating}, there may remain imbalance that needs to be adjusted. So in this paper, we propose an approach that extends the U test reviewed in Section~\ref{sec:background} to be applicable to situations with unbalanced baseline covariates.

\subsection{Notation and setup}
\label{sec:notations}
We introduce some additional notations here. For each stratum $s$, where $s\in\{1,...,S\}$, we use $X_s^t = \{X_{si}^t,i = 1,...,n_s^t\}$ to denote the collection of baseline covariates for subjects in the treatment group where the first element of each vector $X_{si}^t$ is 1, corresponding to an intercept term. Similarly $X_s^c = \{X_{si}^c,i = 1,...,n_s^c\}$ is used to denote the covariates for subjects in the control group. Let 
$X_s = X_s^t \cup X_s^c$ be the collection of covariates for all subjects in stratum $s$, where we assume the first $n_s^t$ subjects are from the treatment group, and the rest are from the control. We use $T_s = \{T_{si}, i = 1,...,n_s\}$ to denote the indicators of treatment, i.e., the first $n_s^t$ elements are 1's and the rest are 0's. The within-stratum propensity score,  $P(T_s = 1 | X_s)$, is denoted by $e(X_s) = \{e(X_{si}),i = 1,...,n_s\}$. Similarly, $e(X_s^t) = \{e(X_{si}^t), i = 1,...,n_s^t\}$ denotes the first $n_s^t$ elements in $e(X_s)$ and $e(X_s^c) = \{e(X_{si}^c), i = 1,...,n_s^c\}$ denotes the rest. We assume $0<e(X_s)<1$ for all $s\in\{1,\cdots, S\}$.

\subsection{Balancing baseline covariates within one stratum}
To balance confounding variables, one way is to weight the subjects such that within each stratum all baseline covariates from the two treatment groups have the same distributions. As we assume the strata are mutually independent, here we only focus on how to balance the covariates in one stratum, and the same process can be applied to the others. For simplicity, here we omit the stratum indicator $s$ in the subscript.  In one stratum, for baseline covariate $X$, let its marginal density function (or probability mass function if $X$ is discrete) be $f(x)$, and its conditional density functions (or probability mass functions) in the treatment and control groups be $f^t(x)$ and $f^c(x)$, respectively. Our goal is to find weight functions, $w^t(x)$ and $w^c(x)$, in the treatment and control group such that $f^t(x)w^t(x) = f^c(x)w^c(x)$.  As discussed in \cite{li2018balancing}, different choices of weight functions will lead to different target populations of interest. They propose to use a general function $h(x)$ to define the population of interest with $h(x) f(x)$ as its marginal distribution. For example, when $h(x) = 1$, the target population has a marginal distribution of $f(x)$, which corresponds to the distribution of $X$ in the combined population of treatment and control groups. When $h(x)$ is $e(x)$ or $1-e(x)$, the target population refers to the subjects in the treatment or control groups. And when $h(x) = e(x)(1-e(x))$, the target population is the so-called overlap population \citep{li2018balancing}.

For a given $h(x)$, the weight functions $w^t(x)$ and $w^c(x)$ should satisfy
\begin{equation}
    w^t(x)f^t(x) \propto w^c(x)f^c(x) \propto f(x)h(x).
    \label{eq:requiredWeight}
\end{equation}
Since $f^t(x) \propto f(x)e(x)$ and $f^c(x) \propto f(x)(1-e(x))$, (\ref{eq:requiredWeight}) implies
\begin{equation}
    w^t(x)\propto \frac{h(x)}{e(x)},  ~~\text{and}~~w^c(x)\propto \frac{h(x)}{1-e(x)}.
    \label{eq:weights}
\end{equation}
When $h(x) = 1$, the induced weight functions yield the classical inverse probability weighting  \citep{horvitz1952generalization}.

The aforementioned weighting method can be incorporated in U-statistics as well. For example, \cite{satten2018multisample} adopted it to adjust two-sample U-statistics with the goal of testing for the existence of treatment effect in observational studies. For our study, we also use this method to adjust the pairwise U-statistics introduced in Section~\ref{sec:background} in order to test for treatment effect heterogeneity in observational studies. We take $U^{(1,2)}$ in equation~(\ref{eq:U12}) as an example, which is the average of several kernel functions. Each kernel function $\phi^{(1,2)}(i,j,k,l)$ is constructed by the outcomes of four independent subjects, and each subject needs to be weighted. Since the outcomes are mutually independent, $\phi^{(1,2)}(i,j,k,l)$ should be weighted by the product of the weights for the four subjects, i.e., $w^t(X_{1i}) \cdot w^c(X_{1j}) \cdot w^t(X_{2k}) \cdot w^c(X_{2l})$.

The choice of the weight functions depends on $h(x)$, which in principle can be chosen as any positive function. However, we further require $h(x)$ to be a constant or a function of $e(x)$, and we require it to be differentiable with respect to $e(x)$. These requirements will later greatly help with the efficient estimation of the asymptotic reference distribution for the adjusted U-statistics without requiring approximation/sampling methods such as bootstrap. In general, the choice of $h(x)$ is flexible. For example, in our simulation study in Section~\ref{sec:Simulation} and the application study on only children's mental health in Section~\ref{sec:casestudy_onechild} , we focus on $h(x) = 1$. In the employment program evaluation study in Section \ref{sec:caseStudy_labor}, we choose $h(x) = e(x)$.

In practice, the propensity scores are unknown, and it is common to use a logistic regression model between treatment indicators and associated covariates $X_s$ for their estimation. Formally, within stratum $s$ $(s\in\{1,\cdots,S\})$, we consider the following model with parameter $\beta_s$,
\begin{equation}
    \log \left(\frac{e(X_{si})}{1-e(X_{si})}\right) = \beta_s^{T}X_{si}, ~~i = 1,...,n_s.
    \label{eq:LogisticModel}
\end{equation}{}
Note the model specification here is flexible and can be extended to include quadratic (or nonlinear) functions of $X_s$ and interaction terms as needed. The model does not impose any assumptions on the response variable, and in practice it is convenient to conduct model diagnostics for \eqref{eq:LogisticModel} based on \citet{austin2008goodness}.  
The estimate of $\beta_s$, denoted by $\hat{\beta}_s$, can be obtained by solving the estimating equation of logistic regression, denoted as  $\sum_{j = 1}^{n_s} S_{sj}(\hat{\beta}_s) = 0$.

As the propensity scores are functions of $\beta_s$, for simplicity, we denote the weights for subjects in the treatment and control groups by $w_{si}^t(\beta_s)$ $(i = 1,...,n_s^t)$ and $w_{si}^c(\beta_s)$ $(i = 1,...,n_s^c)$, respectively for $s\in\{1, \cdots, S\}$. In practice, these weights can be estimated by their plug-in estimates.

\subsection{Testing treatment effect heterogeneity between two strata}

We start with constructing a test statistic that compares the treatment effects between the first two strata. After weighting, the U-statistic in (\ref{eq:U12}) becomes
\begin{align}
 U_a^{(1,2)} &= \frac{\sum_{i=1}^{n_{1}^t}\sum_{j=1}^{n_{1}^c}\sum_{k=1}^{n_{2}^t}\sum_{l=1}^{n_{2}^c} w_{1i}^t(\hat{\beta}_1)w_{1j}^c(\hat{\beta}_1)w_{2k}^t(\hat{\beta}_2)w_{2l}^c(\hat{\beta}_2)\phi^{(1,2)}(i,j,k,l)}{\sum_{i=1}^{n_{1}^t}w_{1i}^t(\hat{\beta}_1) \cdot \sum_{j=1}^{n_{1}^c}w_{1j}^c(\hat{\beta}_1) \cdot
 \sum_{k=1}^{n_{2}^t}w_{2k}^t(\hat{\beta}_2) \cdot
 \sum_{l=1}^{n_{2}^c}w_{2l}^c(\hat{\beta}_2)}
 \nonumber\\
 & = \frac{1}{n_{1}^tn_{1}^cn_{2}^tn_{2}^c}\sum_{i=1}^{n_{1}^t}\sum_{j=1}^{n_{1}^c}\sum_{k=1}^{n_{2}^t}\sum_{l=1}^{n_{2}^c} \frac{ w_{1i}^t(\hat{\beta}_2)w_{1j}^c(\hat{\beta}_1)w_{2k}^t(\hat{\beta}_2)w_{2l}^c(\hat{\beta}_1)\phi^{(1,2)}(i,j,k,l)}{\Bar{w}_1^t(\hat{\beta}_1)\Bar{w}_1^c(\hat{\beta}_1)\Bar{w}_2^t(\hat{\beta}_2)\Bar{w}_2^c(\hat{\beta}_2)},
 \label{eq:Ua12}
\end{align}
where $\Bar{w}_s^\omega(\hat{\beta}_s) = \frac{1}{n_s^\omega}\sum\limits_{i = 1}^{n_s^\omega}  w_{si}^\omega(\hat{\beta}_s)$ for $s\in\{1, \cdots, S\}$ and $\omega \in \{t,c\}$.

Though $U_a^{(1,2)}$ looks like a generalized U-statistic \citep{korolyuk2013theory}, unfortunately it is not, because $\hat{\beta}_1$ and $\hat{\beta}_2$ are functions of all outcomes in the corresponding strata. Therefore the classical projection theorem cannot be directly applied to $U_a^{(1,2)}$. The key observation is that, if we replace $\hat{\beta}_1$ and $\hat{\beta}_2$ by their estimands, $\beta_1$ and $\beta_2$, then we obtain a generalized U-statistic. Moreover, if $\hat{\beta}_1$ and $\hat{\beta}_2$ are consistent estimates, we would expect the asymptotic properties (e.g., normality) of the generalized U-statistics will still hold for our adjusted U-statistic. This is indeed the case by the following theorem. The proof is based on the idea in \cite{satten2018multisample} where they derived the asymptotic normality for adjusted two-sample U-statistics.

\begin{theorem}
Suppose that $\hat{\beta}_1$ and $\hat{\beta}_2$ are consistent estimators for $\beta_1$ and $\beta_2$ and assume that (1) the outcomes $(Y_1^t, Y_2^t, Y_1^c, Y_2^c)$ are mutually independent; (2) there exist positive constants $0<\lambda_s^\omega<1$ for every $s\in\{1,2\}$ and $\omega \in \{t,c\}$ such that $\frac{n_s^\omega}{n_1+n_2}\rightarrow \lambda_s^\omega$ as $n_1+n_2 \rightarrow \infty$ and (3) $0<e(X_s)<1$ for all $s\in\{1,2\}$ where $e(X_s)$ is defined in Section \ref{sec:notations} . Then as $n_1+n_2\rightarrow \infty$, we have
\begin{equation}
    \sqrt{(n_1+n_2)}(U_a^{(1,2)}-\theta_a^{(1,2)}) \stackrel{\mathcal{D}}{\longrightarrow }\mathcal{N}(0, \sigma_{1,2}^2),
    \label{eq:Ua12Asym}
\end{equation}
where $\theta_a^{(1,2)} = \lim_{n_1+n_2\rightarrow\infty}E[U_a^{(1,2)}]$ and \\ $\sigma^2_{1,2} = \lim_{n_1+n_2\rightarrow\infty}\left\{ n_1^t\text{Var}[\eta_1^{t,(1,2)}(Y_1^t)] +  n_1^c\text{Var}[\eta_1^{c,(1,2)}(Y_1^c)] + n_2^t\text{Var}[\eta_2^{t,(1,2)}(Y_2^t)] + n_2^c\text{Var}[\eta_2^{c,(1,2)}(Y_2^c)]\right\}$ is the asymptotic variance of $\sqrt{(n_1+n_2)}U_a^{(1,2)}$, and the $\eta$ functions are defined in the proof below.
\label{thm:Ua12Asymptotics}
\end{theorem}
\begin{proof}
We prove the asymptotic normality of the adjusted U-statistic $U_a^{(1,2)}$ in (\ref{eq:Ua12}) via approximating by four independent sets of $i.i.d.$ random variables. The asymptotic normality then holds by the Central Limit Theorem. This can be directly generalized to any $U_a^{(p,q)}$ with $1\leq p< q\leq S$. For simplicity, we omit the superscript $(1,2)$ in the following proof and use $\approx$ to denote the equalities up to $o_p(n^{-\frac{1}{2}})$, where $n = n_1+n_2$. Some of the notations we use here are similar to what \cite{satten2018multisample} used in their appendix section. Throughout the proof, we use $\text{plim}$ to denote the limit under convergence in probability. 

We set $\theta^* = \frac{1}{n_1^tn_1^cn_2^tn_2^c}E[\sum\limits_{i=1}^{n_{1}^t}\sum\limits_{j=1}^{n_{1}^c}\sum\limits_{k=1}^{n_{2}^t}\sum\limits_{l=1}^{n_{2}^c}w_{1i}^t(\beta_1)w_{1j}^c(\beta_1)w_{2k}^t(\beta_2)w_{2l}^c(\beta_2)\phi(i,j,k,l)]$,\\ $\theta_s^\omega = \text{plim}~\Bar{w}_s^\omega(\beta_s) = \text{plim}~ \frac{1}{n_s^\omega}\sum\limits_{j = 1}^{n_s^\omega}w_{sj}^t(\beta_s)$ ($s\in\{1,2\}$,  $\omega \in \{t,c\}$), and $\theta_a = \frac{\theta^*}{\theta_1^t\theta_1^c\theta_2^t\theta_2^c}$. By first-order Taylor expansion in four variables at $(\theta^*, \theta_1^t, \theta_1^c, \theta_2^t, \theta_2^c)$, we have
\begin{align}
    U_a - \theta_a = & \frac{1}{n_{1}^tn_{1}^cn_{2}^tn_{2}^c}\sum_{i=1}^{n_{1}^t}\sum_{j=1}^{n_{1}^c}\sum_{k=1}^{n_{2}^t}\sum_{l=1}^{n_{2}^c} \frac{ w_{1i}^t(\hat{\beta}_1)w_{1j}^c(\hat{\beta}_1)w_{2k}^t(\hat{\beta}_2)w_{2l}^c(\hat{\beta}_2)\phi(i,j,k,l)}{\Bar{w}_1^t(\hat{\beta}_1)\Bar{w}_1^c(\hat{\beta}_1)\Bar{w}_2^t(\hat{\beta}_2)\Bar{w}_2^c(\hat{\beta}_2)} - \frac{\theta^*}{\theta_1^t\theta_1^c\theta_2^t\theta_2^c} \nonumber\\
     \approx & c_{11}^t (\Bar{w}_1^t(\hat{\beta}_1) - \theta_1^t) +  c_{11}^c (\Bar{w}_1^c(\hat{\beta}_1) - \theta_1^c) + 
     c_{12}^t (\Bar{w}_2^t(\hat{\beta}_2) - \theta_2^t) + 
      c_{12}^c (\Bar{w}_2^c(\hat{\beta}_2) - \theta_2^c) \nonumber \\
      & +c_2[\frac{1}{n_{1}^tn_{1}^cn_{2}^tn_{2}^c}\sum_{i=1}^{n_{1}^t}\sum_{j=1}^{n_{1}^c}\sum_{k=1}^{n_{2}^t}\sum_{l=1}^{n_{2}^c}w_{1i}^t(\hat{\beta}_1)w_{1j}^c(\hat{\beta}_1)w_{2k}^t(\hat{\beta}_2)w_{2l}^c(\hat{\beta}_2)\phi(i,j,k,l) - \theta^*] 
      \label{eq:AppendixGeneralExpan}
\end{align}{}
where $c_{1s}^\omega = -\frac{\theta_a}{\theta_s^\omega}$ for $s\in \{1,2\}$ and $\omega \in \{t,c\}$, $c_2 = \frac{1}{\theta_1^t\theta_1^c\theta_2^t\theta_2^c}$.

Then by first-order Taylor expansion again, we have
\begin{align}
    \Bar{w}_s^\omega(\hat{\beta}_s) - \theta_s^\omega &= \frac{1}{n_s^\omega}\sum\limits_{i = 1}^{n_s^\omega}w_{si}^\omega(\hat{\beta}_s) - \theta_s^\omega \nonumber \\
     &= \frac{1}{n_s^\omega}\sum\limits_{i = 1}^{n_s^\omega}w_{si}^\omega(\hat{\beta}_s) - \frac{1}{n_s^\omega}\sum\limits_{i = 1}^{n_s^\omega}w_{si}^\omega(\beta_s) + \frac{1}{n_s^\omega}\sum\limits_{i = 1}^{n_s^\omega}w_{si}^\omega(\beta_s) - \theta_s^\omega \nonumber \\
     & \approx c_{3s}^\omega(\hat{\beta}_s-\beta_s) + \frac{1}{n_s^\omega}\sum\limits_{i = 1}^{n_s^\omega}w_{si}^\omega(\beta_s) - \theta_s^\omega~~~~\text{for } s = 1,2; \omega = t,c
\end{align}
where $c_{3s}^\omega = \text{plim}~\frac{1}{n_s^\omega}\sum\limits_{i = 1}^{n_s^\omega}\frac{\partial w_{si}^\omega(\beta_s)}{\partial \beta_s}$.

As $\hat{\beta}_1$ and $\hat{\beta}_2$ are obtained by solving estimating equations $\sum\limits_{i = 1}^{n_1}S_{1j}(\hat{\beta}_1) = 0$ and $\sum\limits_{i = 1}^{n_2}S_{2j}(\hat{\beta}_2) = 0$ respectively, again via first-order Taylor expansion,
\begin{align}
    \hat{\beta}_s-\beta_s\approx -J_s^{-1}\frac{1}{n_s}
    \sum_{i = 1}^{n_s}S_{si}(\beta_s)~~~~\text{for}~~ s = 1, 2
\end{align}{}
where $J_s = \text{plim} ~ \frac{1}{n_s}\sum\limits_{j = 1}^{n_s}\frac{\partial S_{sj}(\beta_s)}{\partial \beta_s}$.
For the last component of (\ref{eq:AppendixGeneralExpan}), by first-order Taylor expansion in two variable at the point $(\beta_1, \beta_2)$, we have
\begin{align}
    &\frac{1}{n_{1}^tn_{1}^cn_{2}^tn_{2}^c}\sum_{i=1}^{n_{1}^t}\sum_{j=1}^{n_{1}^c}\sum_{k=1}^{n_{2}^t}\sum_{l=1}^{n_{2}^c}w_{1i}^t(\hat{\beta}_1)w_{1j}^c(\hat{\beta}_1)w_{2k}^t(\hat{\beta}_2)w_{2l}^c(\hat{\beta}_2)\phi(i,j,k,l) - \theta^* \nonumber\\
    \approx & c_{41} (\hat{\beta}_1 - \beta_1) + c_{42}(\hat{\beta}_2 - \beta_2) + \nonumber \\
    & +\frac{1}{n_{1}^tn_{1}^cn_{2}^tn_{2}^c}\sum_{i=1}^{n_{1}^t}\sum_{j=1}^{n_{1}^c}\sum_{k=1}^{n_{2}^t}\sum_{l=1}^{n_{2}^c}w_{1i}^t(\beta_1)w_{1j}^c(\beta_1)w_{2k}^t(\beta_2)w_{2l}^c(\beta_2)\phi(i,j,k,l) - \theta^*,
    \label{eq:AppendixUstat}
\end{align}{}
where 
\begin{align}
    c_{4s} = \text{plim}~ \frac{1}{n_{1}^tn_{1}^cn_{2}^tn_{2}^c}\sum\limits_{i=1}^{n_{1}^t}\sum\limits_{j=1}^{n_{1}^c}\sum\limits_{k=1}^{n_{2}^t}\sum\limits_{l=1}^{n_{2}^c} \frac{\partial w_{1i}^t(\beta_1)w_{1j}^c(\beta_1)w_{2k}^t(\beta_2)w_{2l}^c(\beta_2)}{\partial \beta_s}\phi(i,j,k,l),~~ s = 1,2
\end{align}{}
Note in (\ref{eq:AppendixUstat}), $\frac{1}{n_{1}^tn_{1}^cn_{2}^tn_{2}^c}\sum_{i=1}^{n_{1}^t}\sum_{j=1}^{n_{1}^c}\sum_{k=1}^{n_{2}^t}\sum_{l=1}^{n_{2}^c}w_{1i}^t(\beta_1)w_{1j}^c(\beta_1)w_{2k}^t(\beta_2)w_{2l}^c(\beta_2)\phi(i,j,k,l)$ is a 4-sample generalized U-statistic with kernel function $\Tilde{\phi}(i,j,k,l) = w_{1i}^t(\beta_1)w_{1j}^c(\beta_1)w_{2k}^t(\beta_2)w_{2l}^c(\beta_2)\phi(i,j,k,l)$. So by the classical projection theorem \citep{hajek1968asymptotic, van2000asymptotic}, we have
\begin{align*}
    &\frac{1}{n_{1}^tn_{1}^cn_{2}^tn_{2}^c}\sum_{i=1}^{n_{1}^t}\sum_{j=1}^{n_{1}^c}\sum_{k=1}^{n_{2}^t}\sum_{l=1}^{n_{2}^c}w_{1i}^t(\beta_1)w_{1j}^c(\beta_1)w_{2k}^t(\beta_2)w_{2l}^c(\beta_2)\phi(i,j,k,l) - \theta^* \nonumber \\
    \approx & \frac{1}{n_1^t}\sum_{i = 1}^{n_1^t}\Tilde{h}_1^t(Y_{1i}^t) + \frac{1}{n_1^c}\sum_{i = 1}^{n_1^c}\Tilde{h}_1^c(Y_{1i}^c) + \frac{1}{n_2^t}\sum_{i = 1}^{n_2^t}\Tilde{h}_2^t(Y_{2i}^t) + \frac{1}{n_2^c}\sum_{i = 1}^{n_2^c}\Tilde{h}_2^c(Y_{2i}^c)-4\theta^*,
\end{align*}{}
 where $\Tilde{h}_s^\omega(x) = E[\Tilde{\phi}(1,1,1,1)|Y_{s1}^\omega = x]$ for $s\in\{1,2\}$ and $\omega \in\{t,c\}$.
 
 Finally, back to Equation (\ref{eq:AppendixGeneralExpan}), we have
 \begin{align}
     U_a - \theta_a & \approx \frac{c_{11}^t}{n_1^t}\sum_{i =1}^{n_1^t}[w_{1i}^t(\beta_1)-\theta_1^t] + \frac{c_{11}^c}{n_1^c}\sum_{i =1}^{n_1^c}[w_{1i}^c(\beta_1)-\theta_1^c] + \frac{c_{12}^t}{n_2^t}\sum_{i =1}^{n_2^t}[w_{2i}^t(\beta_2)-\theta_2^t] + \frac{c_{12}^c}{n_2^c}\sum_{i =1}^{n_2^c}[w_{2i}^c(\beta_2)-\theta_2^c] \nonumber \\
      & - (c_{11}^tc_{31}^t+c_{11}^cc_{31}^c+c_2c_{41})J_1^{-1}\frac{1}{n_1}\sum_{i = 1}^{n_1}S_{1i}(\beta_1) - (c_{12}^tc_{32}^t+c_{12}^cc_{32}^c+c_2c_{42})J_2^{-1}\frac{1}{n_2}\sum_{i = 1}^{n_2}S_{2i}(\beta_2) \nonumber\\
      & + \frac{c_2}{n_1^t}\sum_{i = 1}^{n_1^t}\Tilde{h}_1^t(Y_{1i}^t)-c_2\theta^* + \frac{c_2}{n_1^c}\sum_{i = 1}^{n_1^c}\Tilde{h}_1^c(Y_{1i}^c)-c_2\theta^* + \frac{c_2}{n_2^t}\sum_{i = 1}^{n_2^t}\Tilde{h}_2^t(Y_{2i}^t)-c_2\theta^* + \frac{c_2}{n_2^c}\sum_{i = 1}^{n_2^c}\Tilde{h}_2^c(Y_{2i}^c)-c_2\theta^*\nonumber \\
      & \stackrel{\Delta}{=} \sum_{i = 1}^{n_1^t} \eta_{1}^t(Y_{1i}^t) + \sum_{i = 1}^{n_1^c}\eta_{1}^c(Y_{1i}^c) + \sum_{i = 1}^{n_2^t} \eta_{2}^t(Y_{2i}^t) + \sum_{i = 1}^{n_2^c}\eta_{2}^c(Y_{2i}^c),
      \label{eq:AppendixEtaSum}
 \end{align}{}
 where
 \begin{align}
     \eta_{1i}^t = \frac{c_{11}^t}{n_1^t}[w_{1i}^t(\beta_1)-\theta_1^t] + \frac{c_2}{n_1^t}[\Tilde{h}_1^t(Y_{1i}^t)-\theta^*] - (c_{11}^tc_{31}^t+c_{11}^cc_{31}^c+c_2c_{41})J_1^{-1}\frac{1}{n_1}S_{1i}^t(\beta_1), \text{ for } i = 1,...,n_1^t, \nonumber\\
     \eta_{1i}^c = \frac{c_{11}^c}{n_1^c}[w_{1i}^c(\beta_1)-\theta_1^c] + \frac{c_2}{n_1^c}[\Tilde{h}_1^c(Y_{1i}^c)-\theta^*] - (c_{11}^tc_{31}^t+c_{11}^cc_{31}^c+c_2c_{41})J_1^{-1}\frac{1}{n_1}S_{1i}^c(\beta_1), \text{ for } i = 1,...,n_1^c, \nonumber\\
     \eta_{2i}^t = \frac{c_{12}^t}{n_2^t}[w_{2i}^t(\beta_2)-\theta_2^t] + \frac{c_2}{n_2^t}[\Tilde{h}_2^t(Y_{2i}^t)-\theta^*] - (c_{12}^tc_{32}^t+c_{12}^cc_{32}^c+c_2c_{42})J_2^{-1}\frac{1}{n_2}S_{2i}^t(\beta_2), \text{ for } i = 1,...,n_2^t, \nonumber\\
     \eta_{2i}^c = \frac{c_{12}^c}{n_2^c}[w_{2i}^c(\beta_2)-\theta_2^c] + \frac{c_2}{n_2^c}[\Tilde{h}_2^c(Y_{2i}^c)-\theta^*] - (c_{12}^tc_{32}^t+c_{12}^cc_{32}^c+c_2c_{42})J_2^{-1}\frac{1}{n_2}S_{2i}^c(\beta_2), \text{ for } i = 1,...,n_2^c.
     \label{eq:Etas}
\end{align}{}
As we always assume in each stratum $s$ $(s\in\{1,2\})$, the first $n_s^t$ subjects are in the treatment group, and the last $n_s^c$ subjects are in the control group, here $\{S_{si}^t, i = 1,...,n_s^t\}$ are the first $n_s^t$ elements of $\{S_{si},i = 1,...,n_s\}$, and $\{S_{si}^c, i = 1,...,n_s^c\}$ are the rest elements of it. Since the expectation of the right hand side of (\ref{eq:AppendixEtaSum}) is 0, the limit expectation of $U_a$ is $\theta_a$. By the Central Limit Theorem, Theorem ~\ref{thm:Ua12Asymptotics} is obtained.
\end{proof}

Theorem~\ref{thm:Ua12Asymptotics} establishes the asymptotic distribution for our proposed adjusted U-statistic. Assumptions (1)--(3) are mild and commonly used in the literature. For example, Assumption (2) requires that within stratum, the proportion of treatment/group is not negligible, which is satisfied in most applications. Assumption (3) requires the propensity score to be bounded away from 0 and 1, which is called probabilistic assignment and is commonly used in the causal inference literature \citep{imbens2015causal}.

To estimate the asymptotic variance $\sigma_{1,2}^2$, we first estimate the $\eta$ values, denoted by $\{\Hat{\eta}_s^{\omega, (1,2)}(Y_{si}^\omega), i = 1, \cdots, n_s^\omega\}$ for $s\in\{1,2\}$, by replacing all $\beta$'s by their consistent estimates and replacing $\title{h}$ functions by their method-of-moment estimators in the same way as discussed in Section~\ref{sec:background}. Then we use the sample variance of each set $\{\Hat{\eta}_s^{\omega, (1,2)}(Y_{si}^\omega), i = 1, \cdots, n_s^\omega\}$ to estimate $\text{Var}[\hat{\eta}_s^{\omega,(1,2)}(Y_s^\omega)]$, i.e., $\widehat{\text{Var}}[\hat{\eta}_s^{\omega,(1,2)}(Y_s^\omega)]=\frac{1}{n_s^\omega-1}\sum\limits_{i = 1}^{n_s^\omega}(\hat{\eta}_{s}^{\omega,(1,2)} (Y_{si}^\omega) - \bar{\hat{\eta}}_{s}^{\omega,(1,2)}(Y_s^\omega))^2$, where $\bar{\hat{\eta}}_{s}^{\omega,(1,2)}(Y_s^\omega)$ is the average of $\{\Hat{\eta}_s^{\omega, (1,2)}(Y_{si}^\omega), i = 1, \cdots, n_s^\omega\}$. Then $\sigma^2_{1,2}$ can be consistently estimated by $\widehat{\sigma}^2_{1,2} =n_1^t\widehat{\text{Var}}[\hat{\eta}_1^{t,(1,2)}(Y_1^t)]+  n_1^c\widehat{\text{Var}}[\hat{\eta}_1^{c,(1,2)}(Y_1^c)] + n_2^t\widehat{\text{Var}}[\hat{\eta}_2^{t,(1,2)}(Y_2^t)] +  n_2^c\widehat{\text{Var}}[\hat{\eta}_2^{c,(1,2)}(Y_2^c)] $. 


\subsection{Testing treatment effect heterogeneity in multiple strata}

Next we consider testing for treatment effect heterogeneity in multiple strata, i.e., $1,2,\cdots,S$, with $S>2$, by extending the adjusted U-statistic in the previous section. For every pair of strata $p$ and $q$ satisfying $1\leq p<q\leq S$, we can define an adjusted U-statistic $U_a^{(p,q)}$ in the same way as $U_a^{(1,2)}$. Then it is natural to consider a vector of all pairwise adjusted U-statistics $U_a = (U_a^{(1,2)}, U_a^{(1,3)},...,U_a^{(S-1,S)})^T$. In the next theorem, we derive its joint asymptotic distribution.

\begin{theorem}
Suppose that Assumptions (1)--(3) in Theorem \ref{thm:Ua12Asymptotics} are satisfied for every stratum. Then as the total sample size $N\rightarrow\infty$,
\begin{equation}
\sqrt{N} (U_a-\theta_a)
  \stackrel{D}{\longrightarrow} \mathcal{N}(0,\Sigma_a),
 \label{eq:Ua_asymptotics}
\end{equation}{}
where $\theta_a =\lim_{N\rightarrow\infty} E(U_a)$ and 
$\Sigma_a = \frac{1}{\lambda_1^t}\Sigma_1^t + 
    \frac{1}{\lambda_1^c}\Sigma_1^c+...+ \frac{1}{\lambda_S^t}\Sigma_S^t + 
    \frac{1}{\lambda_S^c}\Sigma_S^c$
is the asymptotic covariance matrix of $\sqrt{N}U_a$,    $\Sigma_s^\omega$ is the covariance matrix of $(\Tilde{\eta}_s^{\omega,(1,2)},..., \Tilde{\eta}_s^{\omega,(S-1,S)})$ for $s\in\{1,\cdots,S\}$ and $\omega\in\{t,c\}$, where $\Tilde{\eta}_s^{\omega,(p,q)} = \eta_s^{\omega,(p,q)} I(s=p \text{ or } s=q)$.
\label{thm:UaVecAsymptotics}
\end{theorem}

\begin{proof}
Following the proof of Theorem \ref{thm:Ua12Asymptotics} for $U_a^{(1,2)}$, we define $\hat{U}_a^{(1,2)}$ as 
\begin{equation}
    \hat{U}_a^{(1,2)} = \sum_{i = 1}^{n_1^t} \eta_{1}^{t,(1,2)}(Y_{1i}^t) + \sum_{i = 1}^{n_1^c}\eta_{1}^{c,(1,2)}(Y_{1i}^c) + \sum_{i = 1}^{n_2^t} \eta_{2}^{t,(1,2)}(Y_{2i}^t) + \sum_{i = 1}^{n_2^c}\eta_{2}^{c,(1,2)}(Y_{2i}^c).
    \label{eq:Ua12Hat}
\end{equation}
Thus we have 
\begin{equation}
    \sqrt{n_1+n_2}(U_a^{(1,2)} - \theta^{(1,2)}_a - \hat{U}_a^{(1,2)}) \stackrel{P}{\longrightarrow} 0, ~~ \text{as } (n_1+n_2) \rightarrow \infty.
    \label{eq:Ua12Hat_asymp}
\end{equation}
For each $U_a^{(p,q)}$ with $1\leq p<q\leq S$, we have $\hat{U}_a^{(p,q)}$ with the same form of (\ref{eq:Ua12Hat}) satisfying (\ref{eq:Ua12Hat_asymp}). Specifically,
\begin{align}
    &\hat{U}_a^{(p,q)} = \sum\limits_{j = 1}^{n_p^t}\eta_p^{t,(p,q)}(Y_{pj}^t) + \sum\limits_{j = 1}^{n_p^c}\eta_p^{c,(p,q)}(Y_{pj}^c) + \sum\limits_{j = 1}^{n_q^t}\eta_q^{t,(p,q)}(Y_{qj}^t) + \sum\limits_{j = 1}^{n_q^c}\eta_q^{c,(p,q)}(Y_{qj}^c)\\
    &\sqrt{n_p+n_q}(U_a^{(p,q)} - \theta_a^{(p,q)} - \hat{U}_a^{(p,q)}) \stackrel{P}{\longrightarrow} 0, ~~\text{as } (n_p+n_q) \rightarrow \infty.
\end{align}
Under the assumption that $\frac{n_s^\omega}{N}\rightarrow \lambda_s^\omega$ $(0<\lambda_s^\omega<1)$ when $N\rightarrow\infty$, for $s\in\{1,\cdots,S\}$ and $\omega\in\{t,c\}$, we have
\begin{equation}
\sqrt{N}
    \left({\begin{array}{c}
  U_a^{(1,2)}-\theta_a^{(1,2)}-\hat{U}_a^{(1,2)} \\
  U_a^{(1,3)}-\theta_a^{(1,3)}-\hat{U}_a^{(1,3)}\\
  \vdots\\
   U_a^{(S-1,S)}-\theta_a^{(S-1,S)}-\hat{U}_a^{(S-1,S)}
  \end{array} } \right)
  \stackrel{p}{\longrightarrow} 0, \text{ as } N\rightarrow\infty,
\end{equation}{}
and by the multivariate Central Limit Theorem,
\begin{equation}
    \sqrt{N}
    \left(
    {\begin{array}{c}
        \hat{U}_a^{(1,2)}  \\
         \hat{U}_a^{(1,3)}\\
         \vdots \\
         \hat{U}_a^{(S-1,S)}
    \end{array}{}}
    \right)
    \stackrel{D}{\longrightarrow} \mathcal{N}(0,\Sigma_a),
\end{equation}{}
where $\Sigma_a = \frac{1}{\lambda_1^t}\Sigma_1^t + 
    \frac{1}{\lambda_1^c}\Sigma_1^c+...+ \frac{1}{\lambda_S^t}\Sigma_S^t + 
    \frac{1}{\lambda_S^c}\Sigma_S^c$, and $\Sigma_s^\omega$ is the covariance matrix of $(\Tilde{\eta}_s^{\omega,(1,2)},..., \Tilde{\eta}_s^{\omega,(S-1,S)})$ with 
\begin{equation*}
  \Tilde{\eta}_s^{\omega,(p,q)} = 
    \begin{cases}
     \eta_s^{\omega,(p,q)} & \text{ if } s=p \text{ or } s=q, \\
     0 & \text{o.w.} \\
    \end{cases}
\end{equation*}

Therefore we have
\begin{equation}
\sqrt{N} (U_a-\theta_a)
  \stackrel{D}{\longrightarrow} \mathcal{N}(0,\Sigma_a) \text{ as } N\rightarrow\infty.
\end{equation}{}
Theorem~\ref{thm:UaVecAsymptotics} is obtained.
\end{proof}

The asymptotic covariance matrix $\Sigma_a$ in Theorem~\ref{thm:UaVecAsymptotics} can be conveniently estimated in a similar way as for the univariate case in Theorem~\ref{thm:Ua12Asymptotics}. That is, we first estimate the $\eta$ terms and $\Tilde{\eta}$ functions,
and then use the sample covariance matrix of estimated $\Tilde{\eta}_s^{\omega,(p,q)}$ to estimate $\Sigma_s^\omega$ for $s\in\{1,...,S\}$ and $\omega \in \{t,c\}$.

Given the estimated covariance $\hat{\Sigma}_a$, we can construct a global test statistic by considering a transformation on $U_a$. For instance, under $H_0: Y^t_s-Y^c_s$ are identically distributed for $s\in\{1, \cdots, S\}$, we have $\theta_a = \frac{1}{2}\mathbf{1}$; therefore a one-dimensional test statistic can be constructed as $T_a = N(U_a-\frac{1}{2}\mathbf{1})^T(U_a-\frac{1}{2}\mathbf{1})$. Though the analytic form of its reference distribution is not available, we can still approximate it via simulations. This can be done by drawing a large number of samples $\{r_1,\cdots, r_L\}$ from $\mathcal{N}(0,\hat{\Sigma}_a)$, and then use $\{||r_1||^2,\cdots, ||r_L||^2\}$ as the empirical reference distribution. Other functions of $U_a$, e.g., $\sqrt{N}\max\limits_{1\leq p<q\leq S}|U^{(p,q)}-\frac{1}{2}|$, can also be used as the global test statistic, whose reference distribution can be approximated by simulations. In the numerical studies, we focus on using $T_a$, and propose to reject the null hypothesis when $T_a$ is greater than or equal to the $100(1-\alpha)th$ percentile of $\{||r_1||^2,\cdots, ||r_L||^2\}$, where $\alpha$ is the significance level.

\subsection{Trimming Sample}
In the causal inference literature, it is common to exclude subjects with estimated propensity scores too close to 0 or 1 \citep{dehejia1999causal,crump2009dealing,imbens2015causal}. This trimming procedure has been shown to effectively improve the covariate balance between different treatment groups for several reasons. One is that those subjects whose true propensity scores that are equal to 0 or 1 should not be used since there are no counterparts in the alternative group. Another reason is that for those subjects whose estimated propensity scores are very close to 0 or 1, their counterparts will be associated with extremely large weights, which will then lead to a large variance for the estimated treatment effects.

There are two popular trimming rules. One is to set a hard threshold for propensity scores to be included in treatment effect estimates, e.g., $[\gamma, 1-\gamma]$ $(0<\gamma<\frac{1}{2})$ \citep{crump2009dealing}, i.e., subjects with propensity scores outside this range should be removed. The other is that we only use the subjects whose propensity scores are within the overlap region \citep{dehejia1999causal}. Specifically, 
we remove all subjects in the control group whose propensity scores are smaller than the minimum propensity score in the treatment group, and remove all subjects in the treatment group whose propensity scores are larger than the maximum propensity score in the control group. In practice, those two rules can be applied simultaneously.

It is worth mentioning that although the trimming procedure in general improves the treatment effect estimation accuracy, the reference population has changed. Hence there is a trade-off. Under this trade-off, people usually still prefer trimming because a reliable estimate for a subpopulation is generally considered more valuable than an estimate for the original population based on extrapolation or with large variance. In the numerical studies, we present both results  with and without trimming to demonstrate the effect of trimming. More specifically, when implementing trimming, we first remove subjects outside of the propensity score overlap region, and then re-run the same propensity score model for the remaining subjects to obtain the weights for our adjusted U tests. We have conducted several numerical experiments and found that the type I error is better controlled with the new propensity scores. Therefore we choose to implement this trimming procedure for all numerical studies in this paper. 

\section{Simulation}
\label{sec:Simulation}

We conduct simulation studies to evaluate the empirical performance of the proposed adjusted U-statistic test and compare it with the likelihood ratio test (LRT) and the U test developed in \cite{dai2020u}. Here, we focus on the case where the target population is the combination of the treatment and control groups, i.e., $h(x) = 1$. 
We consider the adjusted U tests with and without the trimming procedure, and denote them as AUT-T and AUT, respectively. 

\subsection{Implementation Details}
\label{sec:computationalChallenge}
We first discuss the computational implementation of both our proposed U tests and the LRT. The U test statistic in \eqref{eq:Ua_asymptotics} is a function of $S(S-1)/2$ pairwise adjusted U-statistics, and the computation of each adjusted U-statistic can be expensive in simulation studies.
Therefore instead of calculating the complete adjusted U-statistics, we randomly sample some of the $\phi$ functions in each of the adjusted U-statistics. Take $U_a^{(1,2)}$ in \eqref{eq:Ua12} as an example, for each stratum $s \in \{1, 2\}$ and treatment group $\omega\in\{t,c\}$, we randomly choose $M = 1000N$ ($N$ is the total sample size over all strata) subjects with replacement, denoted by $\{(y_{1i}^t,y_{1i}^c,y_{2i}^t,y_{2i}^c), i = 1,\cdots,M\}$. Then we calculate the kernel function $\phi_i$ based on $(y_{1i}^t,y_{1i}^c,y_{2i}^t,y_{2i}^c)$ and use the weighted average of $\{\phi_i, i = 1,\cdots,M\}$ to approximate $U_a^{(1,2)}$. As we also use the weighted kernel functions to estimate $\Tilde{h}_s^{\omega}(Y_{si}^\omega)$ for $i\in\{1,\cdots,n_s^\omega\}$, $s\in\{1, 2\}$ and $\omega\in\{t,c\}$, which are required to obtain $\hat{\Sigma}_a$, we need to make sure that each subject is sampled at least once. This is usually satisfied given a large sampling size $M$, and we redo the sampling process on the rare occasion that this requirement is not met. The sampling size $M = 1000N$ was selected by running a series of different simulation scenarios with $3$ strata and $N$ ranging from 60 to 3000; this choice of $M$ ensured the variance of the approximated test statistic $\frac{T_a}{N} =  \sum\limits_{1\leq p<q\leq S}(U^{(p,q)}-\frac{1}{2})^2$ to be smaller than 0.003. In order to approximate the reference distribution of $\frac{T_a}{N}$, $10^5$ samples $\{r_i, i = 1,\cdots,10^5\}$ are generated independently from the estimated reference distribution $\mathcal{N}(0,\frac{1}{N}\hat{\Sigma}_a)$. Then $\{||r_i||^2, i = 1,\cdots,10^5\}$ are used to obtain the empirical reference distribution $\frac{T_a}{N}$. The sample size of $10^5$ is chosen to ensure that the variance of the $95^{th}$ percentile of $\{||r_i||^2, i = 1,\cdots,10^5\}$ is below 0.0001.

Next we give a brief review of the competitive approach for testing the treatment effect homogeneity, i.e., the LRT proposed by \cite{gail1985testing}. With $S$ strata, they test the null hypothesis that the average treatment effects $\tau_s$ $(s\in\{1, \cdots, S\})$ are the same across all of the strata versus the alternative that at least two of them are unequal. Assuming the treatment effect estimates $\hat{\tau}_s$ $(s\in\{1, \cdots, S\})$ follow normal distributions as $\hat{\tau}_s \stackrel{indep}{\sim} \mathcal{N}(\tau_s, \sigma_s^2)$, 
then a test statistic is constructed as
\begin{equation}
    H = \sum_{s=1}^S(\hat{\tau}_s-\bar{\hat{\tau}})^2/s_s^2 \stackrel{H_0}{\sim} \chi_{S-1}^2,
\end{equation}
where $\bar{\hat{\tau}} = \frac{\sum_{s=1}^S \hat{\tau}_s/s_s^2}{\sum_{s=1}^S 1/s_s^2}$, and $s_s^2$ is a consistent estimator of  $\sigma_s^2$ for $s\in\{1, \cdots, S\}$. For an $\alpha$ level test, we reject the null hypothesis when $H$ is greater than or equal to the $100(1-\alpha)^{\text{th}}$ percentile of $\chi_{S-1}^2$.

In randomized experiments where we can directly compare the outcomes of different treatment groups to estimate the treatment effect, $\hat{\tau}_s$ can be the difference of the outcome averages. In observational studies, a method for estimating $\hat{\tau}_s$ that adjusts for confounding variables should be used. Any methods that can provide normally distributed $\hat{\tau}_s$ and consistent estimator for $\sigma_s^2$ in stratum $s$ for $s\in\{1,\cdots,S\}$ can be used. For instance, when the outcome follows a continuous distribution, a linear regression model between the outcome and the treatment indicator and other confounding variables can be fitted within each stratum. Under the assumption that the outcomes are mutually independent, the normality assumption for $\hat{\tau}_s$ will be satisfied when the stratum sample size $n_s$ goes to infinity. In this simulation, we fit a linear regression in each stratum $s$ $(s\in\{1,\cdots,S\})$ to obtain $\hat{\tau_s}$ and $\hat{\sigma}_s^2$. We focus on the case that $Y^t_s-Y^c_s$ $(s\in\{1, \cdots, S\})$ follow the same distribution up to a location shift. Thus the hypotheses of the adjusted U tests are equivalent to those of the LRT; hence those two tests are directly comparable.

\subsection{Simulation Design}
\label{sec:simulationDesign}
We consider three strata $(S=3)$, where each stratum has the same sample size, i.e., $n_1 = n_2 = n_3 = n$. For each stratum $s$, we generate the data from an outcome model $Y_{s} = 1 + \beta_{s,t}T_{s} +Z_{s}+\epsilon_s$ for $s\in\{1,2,3\}$, where the treatment indicator $T_s \sim \text{Bern}(p_s)$, and the residual terms $\epsilon_s$ follow a common distribution $F_\epsilon$ across all strata. The probability of being assigned to the treatment group $p_s$ is also a function of the confounding variable $Z_s$, for which we assume $\text{logit}(p_s) = \gamma_sZ_s$. In the following simulations, we fix $Z_1\sim \mathcal{N}(0,1)$, $Z_2\sim \mathcal{N}(0,1)$, $Z_3\sim \text{Unif}(-0.5,0.5)$, and choose $\gamma_1 = 1$, $\gamma_2 = -1$, $\gamma_3 = 1$, such that the confounding variables either follow different distributions or satisfy different relationships with the treatment assignment among the three strata. Also, the treatment effects are set as $\beta_{1,t} = 1$, $\beta_{2,t} = 1+\Delta$, $\beta_{3,t} = 1+2\Delta$, where the constant $\Delta$ is treated as the effect size. Note that when $\Delta = 0$, there still exists a treatment effect within each stratum although there is no treatment effect heterogeneity, i.e., the null hypothesis is true. For all of the simulation scenarios, we fix the significance level at 0.05, and repeat the data generating mechanism for $L = 2000$ times to obtain the empirical rejection rates. 

In addition to the simulation design described above, we also consider several other designs with unequal sample size and different error distributions across the three strata. The simulation designs and results are very similar to those in \cite{dai2020u}, so we choose not to present them in this paper.  


\subsection{Simulation results}
We first check the type I error of our proposed adjusted U-test with and without trimming (AUT-T and AUT) when $\Delta = 0$, $n = 200$, $F_\epsilon = \mathcal{N}(0,1)$, and compare them to the unadjusted U test reviewed in Section~\ref{sec:background}. Based on 2000 Monte-Carlo replications, the type I error rates for the AUT-T and AUT are 0.051 and 0.058,  both are very close to the nominal level of 0.05, whereas the unadjusted U test has a rejection rate of 1.000. The invalidity of the unadjusted U test is not surprising, because the unweighted outcome distributions are quite different between treatment and control groups in each stratum, as shown in Figure~\ref{fig:DensityUnadjustOutcomes}. This finding clearly demonstrates the need for confounder adjustment when testing for treatment effect heterogeneity. In Figure \ref{fig:PvalueUvsExpected}, we plot the empirical p-values with the expected uniformly distributed p-values for both the AUT-T and AUT methods. We find that the empirical distribution for the p-values is very close to the uniform distribution under the null hypothesis, which confirms both the validity of the asymptotic null distribution derived in Theorem \ref{thm:UaVecAsymptotics} and the accuracy of random sampling when calculating the test statistics. Compared to AUT, the results for AUT-T is less perfect due to the fact that the population has changed after trimming the propensity score. To demonstrate the effect of trimming, we present the average number of removed subjects for each strata in Table \ref{tab:Valid_NTrimmed} in Supplementary Material, and find that the effect of trimming is minor since less than $8\%$ of the subjects are removed from each stratum. 


\begin{figure}
    \centering
    \includegraphics[width =2.5in]{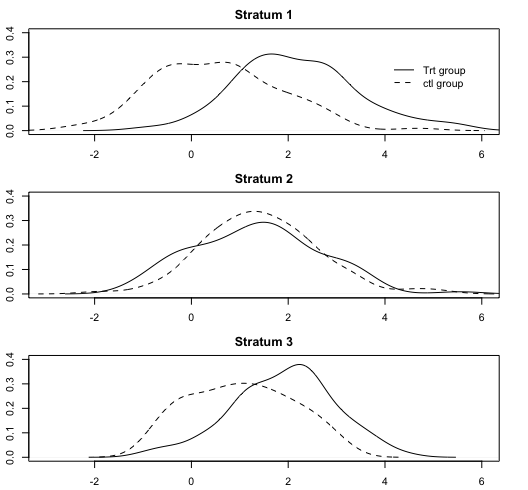}
    \caption{Density plots for the unadjusted outcomes in the treatment and control groups.}
    \label{fig:DensityUnadjustOutcomes}
\end{figure}

\begin{figure}
    \centering
    \includegraphics[width =2.3in]{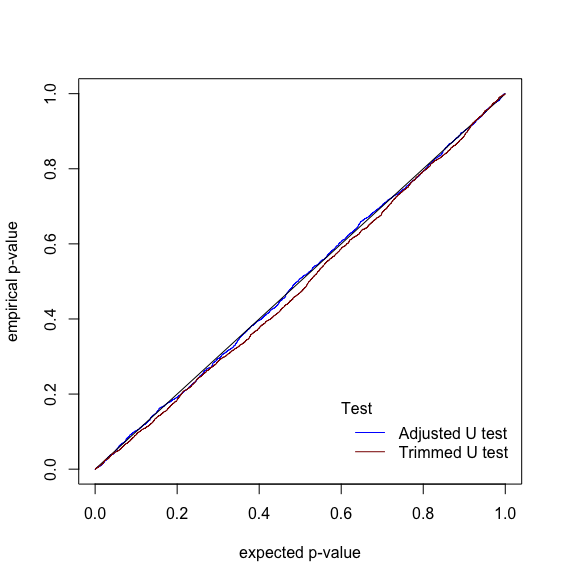}
    \caption{Empirical and expected p-values for proposed U tests under the null hypothesis.}
    \label{fig:PvalueUvsExpected}
\end{figure}

Next we investigate the power for the proposed adjusted U tests under different values for the sample size $n$, effect size $\Delta$ and error distributions $F_\epsilon$. We also use the results from the regression-based LRT as a benchmark for power comparison.

We choose four distributions for $F_\epsilon$: $\mathcal{N}(0,1)$, $\text{Unif}(-2, 2)$, $t_4$ and $0.5\mathcal{N}(-5,1)+0.5\mathcal{N}(5,1)$. For each of them, we consider four effect sizes (including 0), and then present the empirical rejection rates for the adjusted U tests and the LRT in Figure~\ref{fig:PowerVsN}. We first note that under all four scenarios, the type I error rates are very close to the nominal level $0.05$. There is a minor discrepancy for the trimmed U test, especially when the sample size is small. This is expected because trimming changes the reference population, although the number of trimmed subjects (see Figure \ref{fig:PowerVsN_NTrimmed} in Supplementary Material) is quite small (between $2\%$ and $15\%$). Therefore it is fair to compare the power of those three tests given that their type I errors are at the same level.

When $\Delta>0$, we first notice that the power increases quickly to one as either the sample size $n$ or the effect size $\Delta$ increases. By comparing the power between the two adjusted U tests (AUT-T and AUT), we find that overall ATU-T has a larger power although the advantage is not significant. This is expected because we only remove a minor percentage of subjects by trimming. We then compare the power of the AUT and LRT, and find that LRT is more powerful than AUT if the error distribution $F_{\epsilon}$  is normal or having lighter tails than normal distribution (e.g., uniform distribution). On the other hand, our proposed AUT is more powerful than LRT when $F_{\epsilon}$ has heavy tails (e.g., $t_4$) or deviates far away from a normal distribution (e.g., a bimodal distribution as $0.5\mathcal{N}(-5,1)+0.5\mathcal{N}(5,1)$). Those findings confirm that the LRT is still the most powerful test under the normality assumption. However, our proposed method will gain efficiency in testing against the null hypothesis as the true error distribution starts to move away from a normal distribution, with a more significant improvement in power over LRT when the error distribution is bimodal.

\begin{figure}
    \centering
    \includegraphics[width = 2.5in]{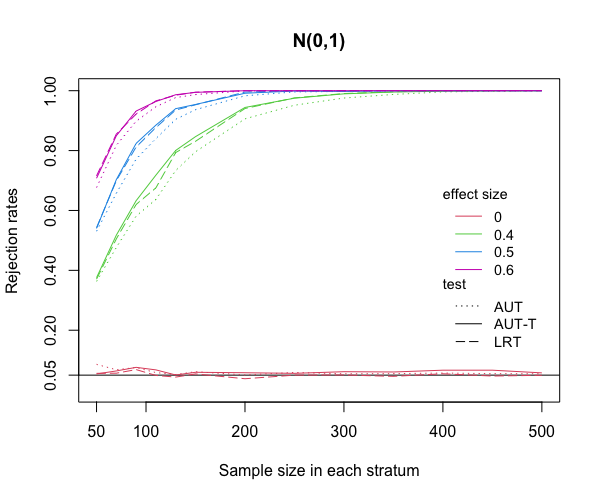}
    \includegraphics[width = 2.5in]{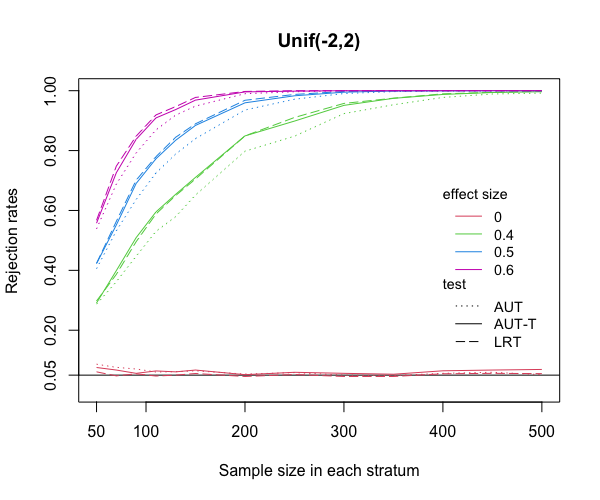}
    \includegraphics[width = 2.5in]{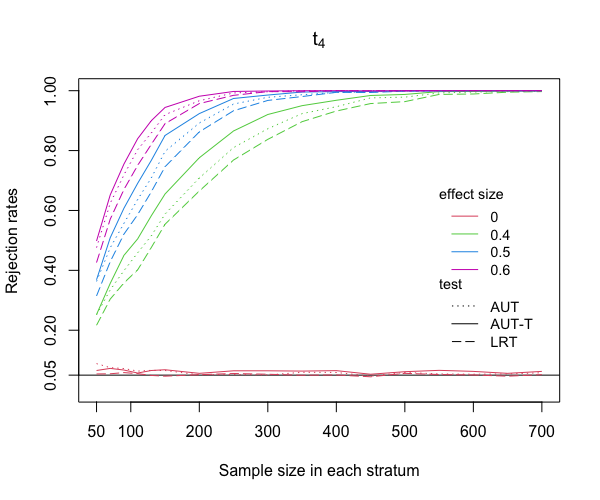}
    \includegraphics[width = 2.5in]{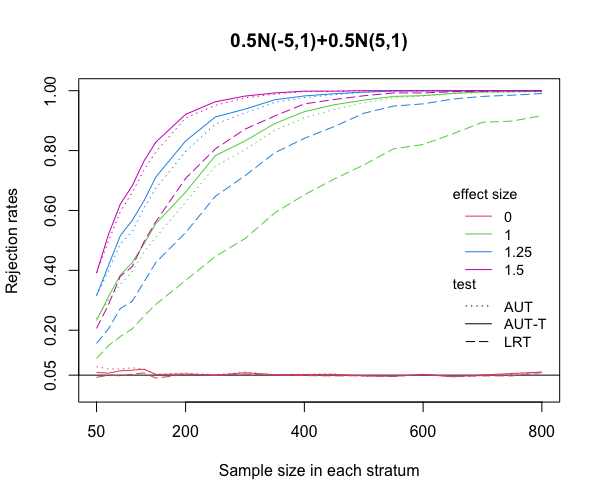}
    \caption{Power analysis: empirical rejection rates for three tests under various error distributions, sample sizes, and effect sizes, based on 2000 Monte-Carlo replications.}
    \label{fig:PowerVsN}
\end{figure}


\subsection{Sensitivity Analysis}
\label{sec:simu_sensitivity}
Because our proposed adjusted U test is based on a propensity score model, in this section, we conduct a sensitivity analysis to evaluate the performance of our method under misspecification of the propensity score model. It is worth mentioning that despite recent advances in propensity score model diagnosis \citep{imbens2015causal,vegetabile2020optimally} by measuring the degree of covariance balance from the weighted samples in the treatment and control groups, measuring covariate balance still remains challenging especially when the number of covariates is large. Therefore it remains important to explore the sensitivity of the proposed test.

\cite{dehejia1999causal} implemented matching and subclassification on propensity scores to estimate treatment effects (in their approach, they removed subjects with estimated propensity scores outside of the overlap region between treatment and control groups), and compared its sensitivity with a linear-regression-based approach. Both of the approaches have one model to specify. The propensity-score-based approach needs to specify the propensity score model, and the linear-regression-based approach needs to specify the outcome model. \cite{dehejia1999causal} found that the misspecification of propensity score model had a smaller impact than that of the outcome model. Motivated by their findings, we compare the sensitivity of our adjusted U tests based on propensity scores to the LRT which uses linear regressions to estimate stratum-specific treatment effects.

For data generation, we consider three strata $(S = 3)$, each with a sample size of $200$, and a confounding variable $Z_s$ $(s = 1,2,3)$ in each stratum satisfying $Z_1\sim \mathcal{N}(0,0.5^2)$, $Z_2\sim \mathcal{N}(0,0.5^2)$, $Z_3\sim \text{Unif}(-2,2)$. We add a quadratic term of $Z_s$ to both the outcome model and propensity score model as $Y_s = T_s+Z_s+\beta_{s,2}Z_s^2+\epsilon_s$ and $\text{logit}(p_s) = \gamma_{s,0}+Z_s+\gamma_{s,2}Z_s^2$  with $T_s\sim \text{Bern}(0,p_s)$ and $\epsilon_s\sim F_\epsilon$ for $s\in\{1,2,3\}$. Note that there is no treatment effect heterogeneity in this scenario, i.e., the null hypothesis is true. Furthermore, we set $\beta_{1,2} = \gamma_{1,2} = 2 $, $ \beta_{2,2} = \gamma_{2,2}=-2$, $\beta_{3,2} = \gamma_{3,2} =2$, $\gamma_{1,0} = -0.5$, $\gamma_{2,0} = 0.5$, $\gamma_{3,0} = 0.5$, to make the confounding variable $Z$ play a similar role in both the outcome and propensity score models (the coefficients for $Z$ and $Z^2$ are the same in both models for every stratum). Here we explore the extent to which the empirical distributions of p-values for the adjusted U tests deviate from the expected p-values when the propensity model is fitted without the quadratic term. Also, for the LRT, we check the empirical distribution of the p-values when the treatment effects are estimated without the quadratic term. Similarly with before, we consider four cases with the error distribution $F_\epsilon$ as $\mathcal{N}(0,1)$, $\text{Unif}(-2, 2)$, $t_4$ and $0.5\mathcal{N}(-5,1) + 0.5\mathcal{N}(5,1)$.


Figure~\ref{fig:misspecify_Pvalue} shows the relationship between the empirical p-values versus the expected uniform p-values for three tests under each of the four error distributions. Among the three tests, the AUT-T is always the most robust to model misspecification, while AUT and LRT have comparable performance, e.g., AUT is slightly more robust when the error distribution is $\mathcal{N}(0,1)$ or $\text{Unif}(-2,2)$, while LRT is more robust when $F_\epsilon$ is $t_4$. This finding confirms the benefit of an increasing level of model robustness offered by subject trimming based on propensity scores, and is consistent with the findings in \cite{dehejia1999causal}. We also show the average number of trimmed subjects for each stratum in Table \ref{tab:sensitivity_NTrimmed} in Supplementary Material, and find that proportion to be reasonably small ($<5$\%).

\begin{figure}
    \centering
    \includegraphics[width = 2.5in]{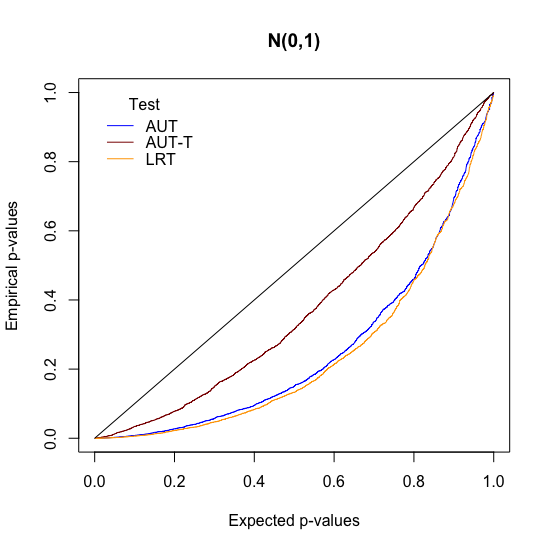}
    \includegraphics[width = 2.5in]{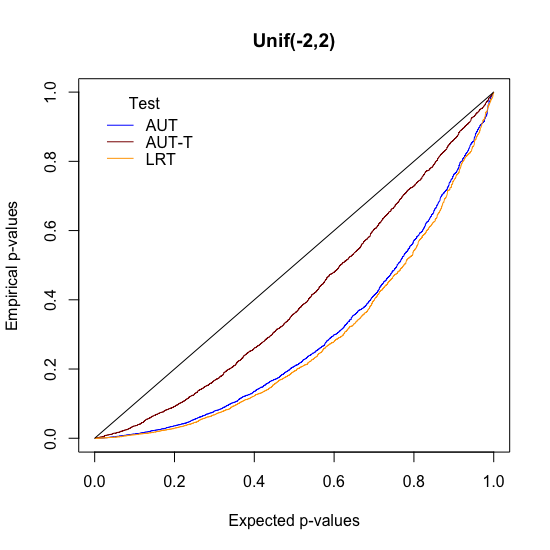}
    \includegraphics[width = 2.5in]{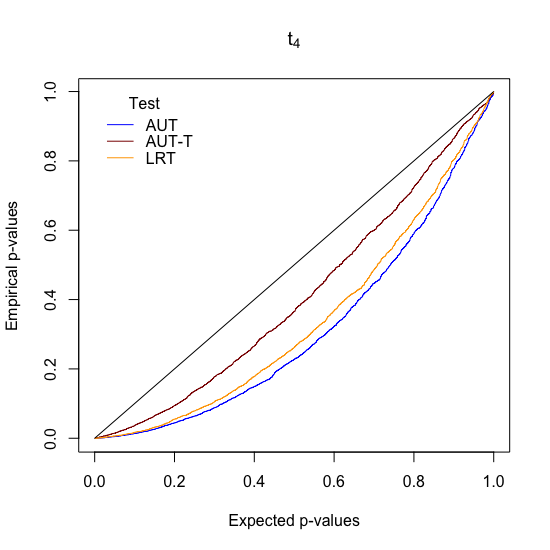}
    \includegraphics[width = 2.5in]{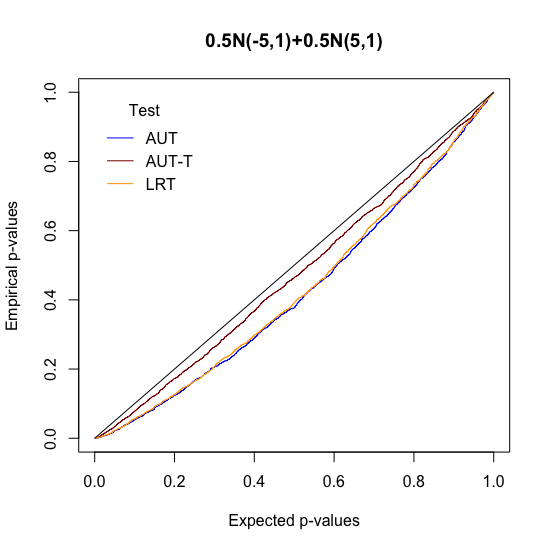}
    \caption{Empirical p-values of misspecified adjusted U test, trimmed U test and LRT vs expected p-values.}
    \label{fig:misspecify_Pvalue}
\end{figure}

\section{Case Study}
\label{sec:caseStudy}

\subsection{Comparing effects of an employment program on people with different ages}
\label{sec:caseStudy_labor}
We apply the proposed method to an employment program evaluation study in labor economics, which evaluates the effect of the National Support Work (NSW) Demonstration on trainee earnings. The NSW was conducted in the mid-1970s with the goal of helping disadvantaged workers gain working experience.  More details about this program can be found in \cite{lalonde1986evaluating} and \cite{dehejia1999causal}. In this program, applicants were randomly assigned to the treatment and control groups; and the treatment effect can be easily assessed by directly comparing the outcomes between those two groups. In order to evaluate whether observational studies can replicate results from randomized experiments, \cite{lalonde1986evaluating} compared the treated subjects in the experiment to two nonexperimental comparison groups, namely, the Panel Study of Income Dynamics (PSID-1) and Current Population Survey-Social Security Administration File (CPS-1), as well as several subsets of them. The collected pretreatment covariates include age, education, marital status, indicator of ``no degree", race indicators, earnings in 1974 (RE74) and 1975 (RE75). The outcome of interest is earnings in 1978.

We focus on the data set constructed by \cite{dehejia1999causal}, which is a subset of the original data set in \cite{lalonde1986evaluating} that includes data collected from male participants who have earnings information in 1974. The data is available at \url{https://users.nber.org/~rdehejia/data/.nswdata2.html}. It has been shown by \cite{dehejia1999causal} that there is a positive treatment effect.  Our goal here is to investigate whether there is treatment effect heterogeneity across different age groups for the treated subjects. Two strata are created based on the median age ($25$ years old) of the treatment group, that is, stratum 1 for subjects with age $\leq 25$ and stratum 2 for age $>25$. Figure~\ref{fig:CaseStudyOutcomeDist} in Supplementary Material shows the outcome distributions of the treated subjects in the two strata, and it is clear that both distributions are highly right-skewed, which suggests that nonparametric U tests should be preferred to the LRT.

We compare the NSW treatment group to the NSW control group and CPS-1 separately. The first three columns of Table~\ref{tab:covariateBalance} shows the summary statistics of baseline covariates in both strata for the three groups. 
To compare the NSW treatment group with its control, we notice that the baseline covariates between groups are similarly distributed, so the unadjusted U test can be applied to assess the treatment effect heterogeneity between the two strata. We obtain an estimated unadjusted U-statistic of 0.554 with a p-value of 0.181, which suggests that the treatment effect in the younger group (stratum 1) is smaller than that in the elder group (stratum 2), although this difference is not statistically significant (note that a U-statistic value of 0.5 means no heterogeneity between those two strata, and a value larger than 0.5 means stratum 1 has a smaller treatment effect than that of stratum 2). 

\begin{table}
\centering
\scalebox{0.76}{
\begin{tabular}{lllll}
  \hline
 & NSW Treated & NSW Control & CPS-1 & Weighted and Trimmed CPS-1 \\ 
  \hline
  \textbf{Stratum 1} \\
Sample size & 106 & 161 & 4676 & 2169 \\ 
 \hline
 Age & 21.09 (2.76) & 20.75 (2.75) & 20.82 (2.82) & 20.97 (2.51) \\ 
  Education & 10.29 (1.77) & 9.93 (1.43) & 11.91 (2.14) & 10.2 (1.54) \\ 
  Black & 0.82 (0.39) & 0.8 (0.4) & 0.08 (0.28) & 0.85 (0.36) \\ 
  Hispanic & 0.08 (0.26) & 0.13 (0.33) & 0.07 (0.26) & 0.06 (0.24) \\ 
  Married & 0.11 (0.32) & 0.09 (0.28) & 0.36 (0.48) & 0.1 (0.3) \\ 
  Nodegree & 0.72 (0.45) & 0.89 (0.3) & 0.34 (0.47) & 0.78 (0.41) \\ 
  RE74 & 2129.02 (4809.7) & 2195.81 (6240.8) & 7044.39 (7156.6) & 1845.71 (4032.9) \\ 
  RE75 & 1215.97 (2140.9) & 1125.32 (3037.3) & 7665.79 (7251.4) & 1068.04 (2379.4) \\ 
  \hline
  \textbf{Stratum 2} \\
Sample size & 79 & 99 & 11316 & 1668 \\ 
  \hline
Age & 32.15 (6.24) & 32.05 (6.24) & 38.35 (8.9) & 32.25 (5.97) \\ 
Education & 10.42 (2.28) & 10.35 (1.84) & 12.07 (3.12) & 10.47 (2.1) \\ 
Black & 0.87 (0.33) & 0.87 (0.33) & 0.07 (0.26) & 0.89 (0.32) \\ 
Hispanic & 0.04 (0.2) & 0.07 (0.26) & 0.07 (0.26) & 0.03 (0.17) \\ 
Married & 0.29 (0.46) & 0.26 (0.44) & 0.86 (0.35) & 0.24 (0.42) \\ 
Nodegree & 0.7 (0.46) & 0.74 (0.44) & 0.28 (0.45) & 0.67 (0.47) \\ 
RE74 & 2050.7 (4957.2) & 1962.64 (4611.6) & 16897.94 (8936.6) & 1993.3 (4772.0) \\ 
RE75 & 1956.17 (4204.0) & 1497.18 (3178.3) & 16123.93 (8876.9) & 1909.62 (4093.3) \\ 
   \hline
\end{tabular}}
\caption{Sample means (standard deviations) of baseline characteristics for NSW and CPS-1 data in two age strata.}
\label{tab:covariateBalance}
\end{table}

We then study the comparison between the NSW treatment group and CPS-1 group.  Table~\ref{tab:covariateBalance} suggests that the baseline covariate distributions in those two groups seem to differ quite a lot. Therefore we apply the proposed adjusted U test with trimming. In both strata, we use logistic regressions to estimate propensity scores. For stratum 1 we use the following covariates: age, $\text{age}^2$, $\text{age}^3$, education, $\text{education}^2$,  $I(\text{married})$, $I(\text{no degree})$, $I(\text{black})$, $I(\text{Hispanic})$, RE74, RE75, $I(\text{RE74} = 0)$, $I(\text{RE75} = 0)$, $\text{RE74} *I(\text{married})$ and  $\text{RE74} * I(\text{no degree})$. In stratum 2, we consider age, $\text{age}^2$, $\text{age}^3$, education, $\text{education}^2$,  $I(\text{married})$ +$I(\text{no degree})$, $I(\text{black})$, $I(\text{Hispanic})$, RE74, RE75, $I(\text{RE74} = 0)$, $I(\text{RE75} = 0)$ and $\text{education} * \text{RE74}$. Most of those covariates are also included in the study of \citet{dehejia1999causal}. Subjects are weighted according to (\ref{eq:weights}) with $h(x) = e(x)$. We present summary statistics for the baseline covariates after trimming and weighting as in the fourth column of Table~\ref{tab:covariateBalance}. The weighted distributions of baseline covariates in CPS-1 are very similar to the NSW treatment group. Due to the large sample size, we randomly sample $M = 1000N$ ($N = 4022$ is the total sample size) weighted kernel functions to approximate the adjusted U-statistics as illustrated in Section~\ref{sec:computationalChallenge}. The estimated adjusted U-statistic comparing the treatment effects in the two strata is 0.541 with a p-value of 0.508, which leads to the same conclusion as the randomized data comparison (NSW treatment versus its control). Meanwhile, if an unadjusted U test is applied to conduct the same comparison, then the estimated U-statistic would be 0.426 with a p-value of 0.004, which will lead to an opposite conclusion. This finding confirms the benefit of our proposed methodology and also highlights the necessity of appropriately adjusting for covariate balance between groups when testing for a treatment heterogeneity effect.

\subsection{Assessing heterogeneity of the effect of being an only child on mental health}
\label{sec:casestudy_onechild}
From 1979 to 2015, China's one-child policy was implemented to slow down the rapid growth of the nation's population. Though the policy has led to economic benefits for China, it has been criticized for introducing a series of social problems, e.g., forced abortions, female infanticide, and a heavy burden of elderly support \citep{hesketh1997one}. Apart from these problems, the psychological wellbeing of the massive number of only children resulting from the policy has been a great concern because it has been widely recognized that siblings have a large impact on children's social behavior and mental health \citeg{dunn1988sibling,mchale2012sibling}. Only children in China are generally perceived to be more self-centered and less trustworthy. 
However the difference between only and non-only children may vary with geographic area and gender for two reasons. First,  parents living in urban and rural areas differ in many aspects including education level, family income and lifestyle. Second, a preference for male children was prevalent at that time, especially in rural areas. For these reasons, the literature assessing the effects of being an only child are typically carried out in different strata that are determined by the type of region (urban/rural) and gender (male/female). For example, \cite{wu2014only} found that only children have worse mental health than children with siblings on average in China, but this negative effect mainly came from rural males,  whereas \cite{zeng2020being} found that the negative effects were more significant in urban areas. It is hence of interest to apply the adjusted U test to study whether there is significant treatment effect heterogeneity among the four subpopulations: urban males, urban females, rural males and rural females.

The data we use is obtained from the Chinese Family Panel Studies (CFPS) \citep{xie2014introduction}, which is a longitudinal survey aiming at documenting changes in various aspects of Chinese society. The baseline survey was conducted in 2010. It covers 25 provinces/municipalities/autonomous regions that represent 95\% of the Chinese population. The data set we focus on is a subset of the CFPS baseline sample constructed by \cite{zeng2020being}. It consists of children born after 1979 with ages between 20 and 31. The data set is available a \url{https://rss.onlinelibrary.wiley.com/pb-assets/hub-assets/rss/Datasets/RSSA%20183.4/A1595Zeng-1600084584507.zip}.
For families with more than one child, only the oldest child is included in the data set. Baseline covariates include age, ethnicity (Han or not), parents' education level (in years), family income in 2010, parents' marital status (divorced or not), parents' ages when the child was born, region type (urban/rural) and gender. The responses include three self-rated psychological measures: confidence, anxiety and desperation. All measures take integer values from 1 to 5, with a higher value indicating better mental health. We treat the only children as the treatment group and the other children as the control group. 
We also remove subjects with obviously erroneous information, e.g., a parent's age below 14 at the time of the child's birth or any response measure outside the range of the scale. Three children with family annual incomes higher than two million Chinese Yuan are removed because these are dramatically larger than the remainder. The final data set has 4187 subjects, with 971 in the treatment group (only children). The distributions of baseline coavariates and outcomes are summarized in the left-hand side of the first panel in Table \ref{tab:Onechild_baselineDist_strata}. We find that parents with only one child have higher average education level and family income. Among only children, there are large proportions of male or urban subjects compared to children with siblings. With respect to the three responses, their summary statistics are very similar between the two treatment groups. Figure \ref{fig:Onechild_responseDist} in Supplementary Material shows the distributions of the three responses in each treatment group. Apart from the fact that every outcome is similarly distributed in the treatment and control groups, they are all heavily left-skewed.

We first apply the weighted version of the Mann-Whitney test introduced by \cite{satten2018multisample} to assess overall average treatment effects with respect to the three outcomes. We standardize all baseline covariates and then fit a logistic regression to estimate propensity scores. After trimming subjects whose estimated propensity scores are outside of the overlap region, we fit the same logistic regression again with the remaining subjects and use the newly estimated propensity scores for weighting. The weights are based on formulas in (\ref{eq:weights}). As we focus on estimating the average treatment effects, we use $h(x) = 1$. Summary statistics for the baseline and response variables after trimming and weighting are presented in the right-hand side of the first panel in Table \ref{tab:Onechild_baselineDist_strata}. There is a clear improvement in covariate balance, though the summary statistics of the responses do not change much. The adjusted U-statistics and corresponding 95\% confidence intervals are given in the first row of Table \ref{tab:oneChild_MWtests}. Here the expectation of the adjusted U-statistic is the probability that outcome in treatment group is smaller than that in control group. Thus a value larger than 0.5 indicates a negative treatment effect, i.e., worse outcomes for only children. The U-statistics show that only children are less confident, less anxious and more desperate than children with siblings, with the 95\% confidence intervals showing that none of these findings are statistically significant.

\begin{table}
\centering
\begin{threeparttable}
\scalebox{0.63}{
\begin{tabular}{|l|ll|ll|}
  \hline
  & \multicolumn{2}{c|}{Unweighted} & \multicolumn{2}{c|}{Trimmed and Weighted}\\
 & Only children & Children with siblings & Only children & Children with siblings \\ 
  \hline
    \hline
    \textbf{All}\tnote{a} &&&&\\
Sample size & 971 & 3216 & 968 & 3216 \\ 
\hline
Baseline covariates &&&&\\
  Maternal education (yrs) & 7.95 (4.28) & 4.19 (4.29) & 4.44 (4.64) & 4.98 (4.49) \\ 
  Paternal education (yrs) & 8.72 (3.98) & 6.41 (4.36) & 6.21 (4.56) & 6.88 (4.36) \\ 
  Age (yrs) & 24.99 (3.38) & 25.19 (3.51) & 25.38 (3.65) & 25.17 (3.50) \\ 
  Han ethnicity & 0.96 (0.20) & 0.89 (0.32) & 0.88 (0.32) & 0.91 (0.30) \\ 
  Family anuual income (Chinese Yuan) & 56957.5 (58152.7) & 37403.1 (44133.1) & 41324.5 (51470.2) & 42793.1 (54362.3) \\ 
  Parental age at birth (yrs) & 26.83 (3.81) & 27.66 (5.11) & 27.92 (5.68) & 27.45 (4.99) \\ 
  Maternal age at birth (yrs) & 25.09 (3.44) & 25.7 (4.54) & 25.9 (4.67) & 25.55 (4.44) \\ 
  Divorce & 0.03 (0.17) & 0.01 (0.10) & 0.01 (0.10) & 0.01 (0.10) \\ 
  Urban area & 0.78 (0.41) & 0.39 (0.49) & 0.43 (0.50) & 0.48 (0.50) \\ 
  Male & 0.59 (0.49) & 0.47 (0.50) & 0.49 (0.50) & 0.50 (0.50) \\ 
  Outcomes &&&&\\
  Confidence & 3.96 (0.92) & 4.02 (0.95) & 3.95 (0.95) & 4.02 (0.94) \\ 
  Anxiety & 4.62 (0.67) & 4.60 (0.69) & 4.63 (0.69) & 4.61 (0.68) \\ 
  Desperation & 4.68 (0.62) & 4.72 (0.61) & 4.69 (0.62) & 4.73 (0.61) \\ 
   \hline
   \hline
  \textbf{Urban males}\tnote{b} &&&&\\
Sample size & 430 & 580 & 423 & 558 \\ 
\hline
Baseline covariates &&&&\\
  Maternal education (yrs) & 8.66 (4.05) & 5.15 (4.37) & 6.58 (4.70) & 6.63 (4.43) \\ 
  Paternal education (yrs) & 9.44 (3.76) & 7.29 (4.37) & 8.05 (4.14) & 8.20 (4.25) \\ 
  Age (yrs) & 25.38 (3.33) & 25.66 (3.53) & 25.62 (3.46) & 25.71 (3.50) \\ 
  Han ethnicity & 0.98 (0.14) & 0.93 (0.24) & 0.96 (0.20) & 0.96 (0.20) \\ 
  Family anuual income (Chinese Yuan) & 59449.4 (60477.7) & 45266.9 (45263.7) & 51219.8 (56528.3) & 54606.3 (64632.9) \\ 
  Paternal age at birth (yrs) & 26.94 (3.48) & 27.90 (4.90) & 27.24 (4.35) & 27.36 (4.52) \\ 
  Maternal age at birth (yrs) & 25.16 (3.23) & 26.29 (4.32) & 25.61 (4.00) & 25.69 (3.76) \\ 
  Divorce & 0.03 (0.17) & 0.01 (0.10) & 0.02 (0.14) & 0.02 (0.14) \\ 
  \hline
  Outcomes &&&&\\
  Confidence & 4.00 (0.92) & 3.96 (0.98) & 3.98 (0.94) & 3.97 (0.93) \\ 
  Anxiety & 4.61 (0.7) & 4.64 (0.62) & 4.63 (0.67) & 4.65 (0.58) \\ 
  Desperation & 4.67 (0.62) & 4.74 (0.58) & 4.67 (0.61) & 4.75 (0.56) \\ 
  \hline
  \hline
  \textbf{Urban females}\tnote{c}  &&&&\\
  Sample size & 331 & 690 & 330 & 634 \\ 
  \hline
  Baseline covariates &&&&\\
  Maternal education (yrs) & 9.05 (3.74) & 5.83 (4.28) & 7.03 (4.41) & 7.2 (4.21) \\ 
  Paternal education (yrs) & 9.54 (3.39) & 7.43 (4.13) & 8.35 (3.74) & 8.42 (3.91) \\ 
  Age (yrs) & 25.01 (3.37) & 25.69 (3.55) & 25.43 (3.47) & 25.44 (3.46) \\ 
  Han ethnicity & 0.95 (0.22) & 0.93 (0.24) & 0.93 (0.26) & 0.93 (0.24) \\ 
  Family anuual income (Chinese Yuan) & 64914.3 (59182.4) & 50261.4 (61045.9) & 55548.9 (51572.5) & 55414.0 (53315.4) \\ 
  Paternal age at birth (yrs) & 27.07 (3.41) & 27.86 (4.99) & 27.03 (3.74) & 27.21 (3.78) \\ 
  Maternal age at birth (yrs) & 25.47 (3.04) & 25.91 (4.19) & 25.33 (3.34) & 25.54 (3.44) \\ 
  Divorce & 0.03 (0.17) & 0.01 (0.10) & 0.02 (0.14) & 0.02 (0.14) \\ 
  \hline
  Outcomes &&&&\\
   Confidence & 3.89 (0.87) & 3.94 (0.92) & 3.87 (0.87) & 3.99 (0.91) \\ 
  Anxiety & 4.67 (0.57) & 4.62 (0.67) & 4.69 (0.55) & 4.61 (0.67) \\ 
  Desperation & 4.68 (0.60) & 4.73 (0.59) & 4.69 (0.57) & 4.73 (0.60) \\ 
  \hline
  \hline
  \textbf{Rural males}\tnote{d} &&&& \\
  Sample size & 146 & 942 & 146 & 927 \\ 
  \hline
  Baseline covariates &&&&\\
  Maternal education (yrs) & 4.92 (4.00) & 2.92 (3.96) & 3.54 (3.82) & 3.24 (4.10) \\ 
  Paternal education (yrs) & 5.94 (3.90) & 5.7 (4.30) & 5.79 (4.00) & 5.73 (4.28) \\ 
  Age (yrs) & 24.13 (3.37) & 24.97 (3.44) & 24.85 (3.56) & 24.83 (3.42) \\ 
  Han ethnicity & 0.92 (0.28) & 0.87 (0.35) & 0.91 (0.28) & 0.88 (0.32) \\ 
  Family anuual income (Chinese Yuan) & 37539.9 (39878.4) & 31437.3 (38706.5) & 34498.6 (32817.1) & 31936.8 (31243.0) \\ 
  Paternal age at birth (yrs) & 25.66 (4.59) & 27.66 (5.28) & 27.02 (5.20) & 27.27 (5.11) \\ 
  Maternal age at birth (yrs) & 24.08 (4.13) & 25.7 (4.79) & 24.93 (4.43) & 25.39 (4.68) \\ 
  Divorce & 0.01 (0.10) & 0.01 (0.10) & 0.01 (0.10) & 0.01 (0.10) \\ 
  \hline
  Outcomes &&&&\\
   Confidence & 4.07 (0.96) & 4.11 (0.93) & 4.06 (0.92) & 4.12 (0.94) \\ 
  Anxiety & 4.57 (0.75) & 4.57 (0.73) & 4.59 (0.77) & 4.58 (0.73) \\ 
  Desperation & 4.73 (0.59) & 4.72 (0.63) & 4.71 (0.59) & 4.72 (0.63) \\ 
  \hline
  \hline
  \textbf{Rural Females}\tnote{e} &&&&\\
  Sample size & 64 & 1004 & 62 & 950 \\ 
  \hline
  Baseline covariates &&&&\\
  Maternal education (yrs) & 4.48 (4.11) & 3.68 (4.05) & 3.55 (3.99) & 3.76 (4.07) \\ 
  Paternal education (yrs) & 6.05 (4.41) & 5.88 (4.34) & 5.83 (4.46) & 5.95 (4.29) \\ 
  Age (yrs) & 24.23 (3.27) & 24.77 (3.48) & 25.18 (3.63) & 24.9 (3.47) \\ 
  Han ethnicity & 0.92 (0.26) & 0.87 (0.33) & 0.91 (0.28) & 0.92 (0.26) \\ 
  Family anuual income (Chinese Yuan) & 43360.5 (59802.2) & 29620.9 (29073.6) & 28334.0 (25023.5) & 30437.5 (29195.6) \\ 
  Paternal age at birth (yrs) & 27.47 (5.22) & 27.40 (5.14) & 27.71 (5.40) & 27.39 (5.15) \\ 
  Maternal age at birth (yrs) & 24.97 (4.41) & 25.23 (4.60) & 25.35 (4.63) & 25.14 (4.54) \\ 
  Divorce & 0.03 (0.17) & 0.01 (0.10) & 0.01 (0.10) & 0.01 (0.10) \\ 
  \hline
  Outcomes &&&&\\
  Confidence & 3.86 (0.95) & 4.03 (0.95) & 3.77 (0.97) & 4.04 (0.95) \\ 
  Anxiety & 4.56 (0.68) & 4.6 (0.69) & 4.55 (0.65) & 4.60 (0.71) \\ 
  Desperation & 4.61 (0.68) & 4.71 (0.62) & 4.64 (0.64) & 4.71 (0.62) \\ 
  \hline
\end{tabular}}
\caption{Unweighted and weighted sample means (standard deviations) of baseline characteristics and responses in treatment and control groups of four strata}
\label{tab:Onechild_baselineDist_strata}
\begin{tablenotes}
\scriptsize
\item [a]P(only child) is modeled by a logistic regression with all baseline covariates, $\text{(maternal age at birth)}^3$, $\text{(paternal age at birth)}^3$, and $\text{(family income)}^3$.
\item [b]P(only child) is modeled by a logistic regression with all baseline covariates, $\text{(maternal age at birth)}^2$, and $\text{(family income)}^2$.
\item [c]P(only child) is modeled by a logistic regression with all baseline covariates, $\text{(maternal age at birth)}^2$, $\text{(paternal age at birth)}^2$, and $\text{(family income)}^2$.
\item [d]P(only child) is modeled by a logistic regression with all baseline covariates, $\text{(maternal age at birth)}^2$, $\text{(family income)}^2$, and divorce * family income.
\item [e]P(only child) is modeled by a logistic regression with all baseline covariates, $\text{(maternal age at birth)}^2$,  $\text{(family income)}^2$, and Han * age.
\end{tablenotes}
\end{threeparttable}
\end{table}

\begin{table}
\centering
\begin{tabular}{llll}
  \hline
Population & confidence & anxiety & desperation \\ 
  \hline
  All & 0.523 (0.486, 0.559) & 0.489 (0.464, 0.515) & 0.519 (0.495, 0.542) \\ 
Urban Males & 0.497 (0.457, 0.538) & 0.500 (0.465, 0.535) & \textbf{0.537 (0.505, 0.569)} \\ 
  Urban Females & \textbf{0.542 (0.503, 0.582)} & 0.480 (0.446, 0.514) & 0.526 (0.494, 0.559) \\ 
  Rural Males & 0.523 (0.471, 0.575) & 0.492 (0.447, 0.537) & 0.507 (0.466, 0.549) \\ 
  Rural Females & \textbf{0.581 (0.511, 0.651)} & 0.536 (0.472, 0.599) & 0.540 (0.481, 0.600) \\ 
   \hline
\end{tabular}
\caption{Adjusted Mann-Whitney test statistics (95\% CI) for different populations with respect to different response measures}
\label{tab:oneChild_MWtests}
\end{table}

We then split the data into four strata based on gender and region type. The sample sizes and distributions of baseline and response variables in the treatment and control groups are summarized in the left-hand sides of the second to fifth panels in Table \ref{tab:Onechild_baselineDist_strata}. It shows that the baseline characteristics vary among strata. For instance, urban parents have higher education levels and incomes than rural parents. The proportion of males are higher among only children than children with siblings, especially in rural areas. With respect to the response variables, there is no obvious difference among these subgroups. Adjusted Mann-Whitney tests are implemented in each strata separately based on the same weighting procedure described above. The baseline covariates are clearly better balanced in all strata. The adjusted Mann-Whitney test statistics and corresponding 95\% confidence intervals are listed in the second to fifth rows of Table \ref{tab:oneChild_MWtests}. Most tests show insignificant results except for testing desperation among urban males, and confidence among urban females and rural females. All these significant results suggest that only children's mental health is worse than children with siblings. Even these should be interpreted with caution given the large number of tests being carried out. The findings here are related but not exactly the same as those reported in  \cite{zeng2020being}, which found significantly negative treatment effects among both urban female and male strata for almost all responses (except for anxiety of urban females). It is worth noting that the statistical significant findings in both papers are close to the boundary of statistical insignificance, e.g., the confidence intervals of our significant tests and the credible intervals in \cite{zeng2020being} are very close to including the null value, $1/2$, in the intervals. 

Interpretation of the results here is challenging due to the number of strata and outcomes. A further challenge is that the results in Table \ref{tab:oneChild_MWtests} suggest similar results across strata in each column but with some attaining significance and others not. It is natural to ask whether these are significant differences across strata (see, e.g.,\cite{gelman2006difference}). The question can be addressed by assessing treatment effect heterogeneity among the four strata.  We implement our proposed adjusted U test and calculate  the test statistic by randomly selecting $M = 1000N$ $(N = 4030)$ kernel terms with replacement as described in Section \ref{sec:computationalChallenge}. The obtained p-values are respectively 0.142, 0.411 and 0.738, for the response variables confidence, anxiety, and desperation, which indicates that there is no significant treatment effect heterogeneity among the four subpopulations for each of the three outcomes. Pairwise tests among the four strata to examine treatment effect heterogeneity regarding the three response variables are also conducted, and the p-values of the 18 tests are almost uniformly distributed, which further demonstrates that there does not appear to be treatment effect heterogeneity across gender and region types.

\section{Discussion}
\label{seq:discussion}
In this paper, we propose a new nonparametric U test for heterogeneity of treatment effects in observational studies. Our method extends the U test in \cite{dai2020u} for randomized experiments to observational studies by adjusting for the confounding variables using propensity score modeling. Our approach is adaptive to various choices of target population, as long as the general function $h(x)$ used to define the target population is a constant or a differentiable function of propensity score. Many target populations of interest in practice satisfy this requirement, including subjects in treatment and control groups combined, treated subjects and subjects under control.

Compared to its parametric counterpart, the LRT, the proposed adjusted U test inherits the advantages of nonparametric tests in terms of weaker modeling assumptions on the outcome, and a significant improvement in power for non-normally distributed data as shown in the numerical studies.  Meanwhile, the desired property of model robustness from other propensity-score-based techniques \citep{dehejia1999causal} also holds for the adjusted U test when subject trimming is implemented. Compared to linear-regression-based approaches, our method is less sensitive to model misspecification. 

Several future working directions remain open. Firstly, we assume that for our method, all confounding variables are observed, which is untestable and may be subject to violation in practice. It will be of interest to conduct a sensitivity analysis to address this issue. Secondly, we assume there are no missing values of the confounding variables. Multiple imputation \citep{schafer1997analysis} can be used to resolve the issue if the values are missing at random. If they are missing not at random, it will be of interest to extend our work based on ideas from \cite{yang2019causal}. Thirdly, the calculation of U-statistics is based on a random sampling procedure over all pairwise comparison between strata for our method. Developing a more efficient sampling method for faster U-statistic computation will be an interesting future working direction. Fourthly, it will be of interest to extend our test statistic for high-dimensional covariates based on the results in \cite{he2021asymptotically}.




\bibhang=1.7pc
\bibsep=2pt
\fontsize{9}{14pt plus.8pt minus .6pt}\selectfont
\renewcommand\bibname{\large \bf References}
\expandafter\ifx\csname
natexlab\endcsname\relax\def\natexlab#1{#1}\fi
\expandafter\ifx\csname url\endcsname\relax
  \def\url#1{\texttt{#1}}\fi
\expandafter\ifx\csname urlprefix\endcsname\relax\def\urlprefix{URL}\fi

 \bibliographystyle{chicago}      
 \bibliography{Bibliography.bib}   











\newpage

\setcounter{page}{1}
\renewcommand{\baselinestretch}{2}
\renewcommand{\thetable}{S\arabic{table}}  
\renewcommand{\thefigure}{S\arabic{figure}}
\renewcommand{\thesection}{S\arabic{section}}


\markboth{\hfill{\footnotesize\rm Maozhu Dai, Weining Shen and Hal S. Stern} \hfill}

\renewcommand{\thefootnote}{}
$\ $\par \fontsize{12}{14pt plus.8pt minus .6pt}\selectfont


\centerline{\Large\bf Supplementary Material for ``Nonparametric Tests for }
\vspace{2pt}
 \centerline{\Large\bf   Treatment Effect Heterogeneity in Observational Studies''}
\vspace{.25cm}
 \centerline{Maozhu Dai, Weining Shen, Hal S. Stern}
\vspace{.4cm}
 \centerline{\it Department of Statistics, University of California, Irvine}
\vspace{.55cm}
\fontsize{9}{11.5pt plus.8pt minus .6pt}\selectfont
\noindent
This supplementary material includes additional plots and tables for the simulations in Section \ref{sec:Simulation} and the data application examples in Section \ref{sec:caseStudy}.
\par
\setcounter{figure}{0} 
\setcounter{table}{0} 
\setcounter{section}{0}

\section{Average number of trimmed subjects for simulations in Section \ref{sec:Simulation}}
\begin{table}[ht]
\centering
\begin{tabular}{|c|c|c|c|c|c|}
  \hline
  \multicolumn{2}{|c|}{Stratum 1} & \multicolumn{2}{c|}{Stratum 2} & \multicolumn{2}{c|}{Stratum 3}\\
  \hline
  Treatment & Control & Treatment & Control & Treatment & Control \\ 
  \hline
  7.11 & 7.09 & 7.30 & 7.14 & 1.64 & 1.62 \\ 
   \hline
\end{tabular}
\caption{Validity: average number of removed subjects in each subgroup for trimmed U test.}
\label{tab:Valid_NTrimmed}
\end{table}

\begin{figure}[ht]
    \centering
    \includegraphics[width = 3in]{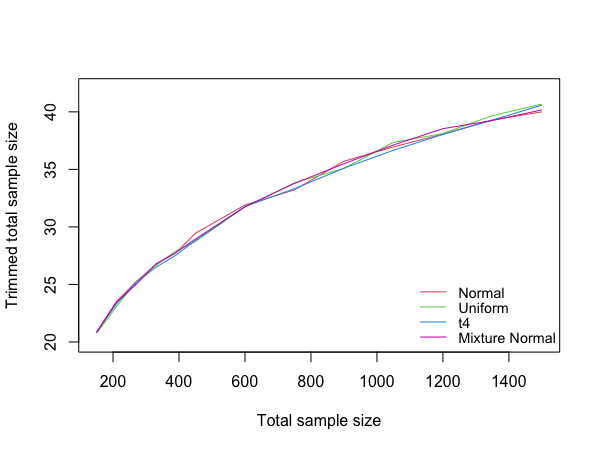}
    \caption{Power analysis: average number of trimmed subjects for four error distributions based on 2000 Monte-Carlo replications.}
    \label{fig:PowerVsN_NTrimmed}
\end{figure}

\begin{table}[H]
\centering
\scalebox{0.8}{
\begin{tabular}{|c|c|c|c|c|c|c|}
  \hline
  &\multicolumn{2}{c|}{Stratum 1} & \multicolumn{2}{c|}{Stratum 2} & \multicolumn{2}{c|}{Stratum 3}\\
  \hline
 & Treatment & Control & Treatment & Control & Treatment & Control \\ 
  \hline
  $\mathcal{N}(0,1)$ & 9.00 & 0.21 & 0.23 & 9.12 & 2.21 & 1.18 \\ 
  $U(-2,2)$ & 8.95 & 0.23 & 0.20 & 8.95 & 2.16 & 1.17 \\ 
  $t_4$ & 8.85 & 0.22 & 0.23 & 9.33 & 2.18 & 1.20 \\ 
  $0.5\mathcal{N}(-5,1)+0.5\mathcal{N}(5,1)$ & 9.03 & 0.19 & 0.21 & 8.99 & 2.20 & 1.19 \\
   \hline
\end{tabular}}
\caption{Sensitivity analysis: average number of trimmed subjects for each stratum by trimmed U test based on 2000 Monte-Carlo replications. }
\label{tab:sensitivity_NTrimmed}
\end{table}

\section{Outcome distributions of applications in Section \ref{sec:caseStudy}}
\label{sec:appendix_figures}
\begin{figure}[ht]
    \centering
    \includegraphics[width = 2.7in]{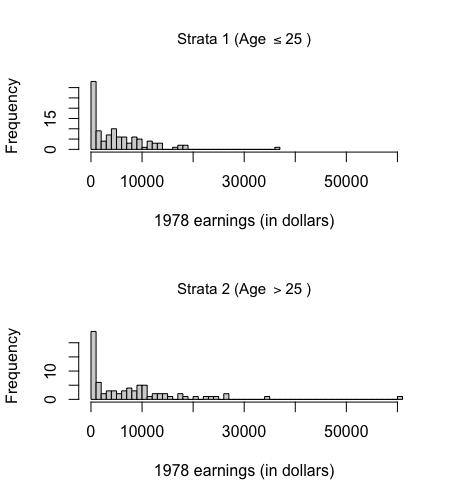}
    \caption{Distribution of earnings in 1974 for participants in the treatment group.}
    \label{fig:CaseStudyOutcomeDist}
\end{figure}

\begin{figure}[ht]
    \centering
    \includegraphics[width = 4in]{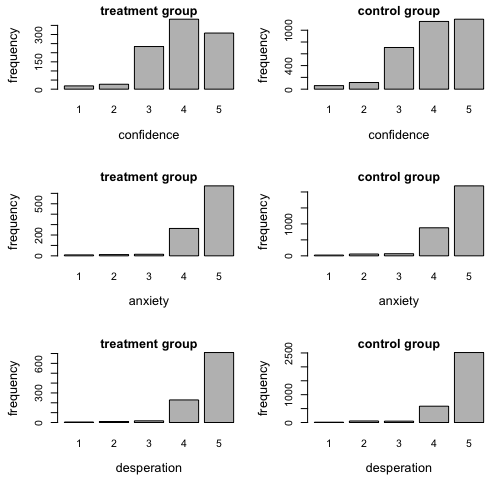}
    \caption{Distributions of confidence, anxiety and desperation measures in the treatment and control groups.}
    \label{fig:Onechild_responseDist}
\end{figure}

\bibhang=1.7pc
\bibsep=2pt
\fontsize{9}{14pt plus.8pt minus .6pt}\selectfont
\renewcommand\bibname{\large \bf References}
\expandafter\ifx\csname
natexlab\endcsname\relax\def\natexlab#1{#1}\fi
\expandafter\ifx\csname url\endcsname\relax
  \def\url#1{\texttt{#1}}\fi
\expandafter\ifx\csname urlprefix\endcsname\relax\def\urlprefix{URL}\fi

\end{document}